\documentclass{article}

\usepackage{algorithm}
\usepackage{url}
\usepackage{wrapfig}

\usepackage[utf8]{inputenc}

\newcommand{\cN}{\mathcal{N}}

\usepackage{microtype}
\usepackage{graphicx}
\usepackage{booktabs} %
\usepackage{natbib}
\usepackage{multirow}

\usepackage{amsmath}
\usepackage{amssymb}
\usepackage{amsthm}
\usepackage{appendix}
\usepackage{bbm}
\usepackage{mathtools}

\newcommand{\norm}[1]{\| #1 \|}
\newcommand{\R}{\mathbb{R}}

\DeclarePairedDelimiterX{\infdivx}[2]{(}{)}{%
  #1\;\delimsize\|\;#2%
}

\usepackage{xcolor}

\usepackage{microtype}
\usepackage{graphicx}
\usepackage{subcaption}
\usepackage{booktabs} %

\usepackage{hyperref}

\usepackage{amsmath}
\usepackage{amssymb}
\usepackage{mathtools}
\usepackage[nice]{nicefrac}
\usepackage{amsthm}

\usepackage[capitalize,noabbrev]{cleveref}

\usepackage[textsize=tiny]{todonotes}

\mathchardef\mhyphen="2D
\newcommand{\Zset}{\ensuremath{\mathcal{Z}}}

\newcommand{\Dfixed}{\ensuremath{D_{\mhyphen}}}
\newcommand{\Dz}{\ensuremath{D_{z}}}

\newcommand{\reconstruct}{\ensuremath{R}}

\newcommand{\lattackloss}{\rho}

\DeclareMathOperator*{\argmax}{arg\,max}
\DeclareMathOperator*{\argmin}{arg\,min}
\DeclareMathOperator*{\argsup}{arg\,sup}
\newcommand{\supp}{\mathop{supp}}
\newcommand{\indicator}{\mathbb{I}}
\newcommand{\one}{\mathbf{1}}
\newcommand{\Ex}{\mathbb{E}}
\renewcommand{\Pr}{\mathbb{P}}
\newcommand{\clip}{\mathrm{clip}}
\newcommand{\blowup}{\mathcal{B}}

\newcounter{thm}
\newtheorem{theorem}[thm]{Theorem}
\newtheorem{lemma}[thm]{Lemma}
\newtheorem{corollary}[thm]{Corollary}
\newtheorem{proposition}[thm]{Proposition}
\newtheorem{definition}[thm]{Definition}
\newtheorem{remark}[thm]{Remark}

\usepackage{algpseudocode}

\newcommand*\samethanks[1][\value{footnote}]{\footnotemark[#1]}
 \newcommand{\Hquad}{\hspace{0.5em}} 

\usepackage[final]{neurips_2023}

\title{Bounding Training Data Reconstruction in DP-SGD}

\author{%
  Jamie Hayes\thanks{Equal contribution.} \\
  Google DeepMind\\
  \texttt{jamhay@google.com} \\
   \And
   Saeed Mahloujifar\samethanks\Hquad\thanks{Work done at Princeton University.} \\
  Meta AI\\
  \texttt{saeedm@meta.com} \\
  \And
   Borja Balle \\
  Google DeepMind\\
  \texttt{bballe@google.com} \\
}

\begin{document}

\maketitle

\begin{abstract}
Differentially private training offers a protection which is usually interpreted as a guarantee against membership inference attacks.
By proxy, this guarantee extends to other threats like reconstruction attacks attempting to extract complete training examples.
Recent works provide evidence that if one does not need to protect against membership attacks but \emph{instead} only wants to protect against training data reconstruction, then utility of private models can be improved because less noise is required to protect against these more ambitious attacks.
We investigate this further in the context of DP-SGD, a standard algorithm for private deep learning, and provide an upper bound on the success of any reconstruction attack against DP-SGD together with an attack that empirically matches the predictions of our bound.
Together, these two results open the door to fine-grained investigations on how to set the privacy parameters of DP-SGD in practice to protect against reconstruction attacks.
Finally, we use our methods to demonstrate that different settings of the DP-SGD parameters leading to the same DP guarantees can result in significantly different success rates for reconstruction, indicating that the DP guarantee alone might not be a good proxy for controlling the protection against reconstruction attacks.

\end{abstract}

\section{Introduction}
\label{sec: introduction}

Machine learning models can and do leak training data~\citep{ippolito2022preventing, kandpal2022deduplicating, carlini2019secret, yeom2018privacy, song2019overlearning, tirumala2022memorization}. 
If the training data contains private or sensitive information, this can lead to information leakage via a variety of different privacy attacks~\citep{balle2022reconstructing, carlini2022membership, fredrikson2015model, carlini2021extracting}.
Perhaps the most commonly studied privacy attack, membership inference~\citep{homer2008resolving, shokri2017membership}, aims to infer if a sample was included in the training set, which can lead to a privacy violation if inclusion in the training set is in and of itself sensitive.
Membership inference leaks a single bit of information about a sample -- whether that sample was or was not in the training set -- and so any mitigation against this attack also defends against attacks that aim to reconstruct more information about a sample, such as training data reconstruction attacks~\citep{balle2022reconstructing, carlini2021extracting, zhu2019deep}.

Differential privacy (DP)~\citep{DBLP:conf/tcc/DworkMNS06} provides an effective mitigation that provably bounds the success of \emph{any} membership inference attack, and so consequently any training data reconstruction attack.
The strength of this mitigation is controlled by a privacy parameter $\epsilon$ which, informally, represents the number of bits that can be leaked about a training data sample, and so $\epsilon$ must be small to guarantee the failure of a membership inference attack~\citep{sablayrolles2019white,DBLP:conf/sp/NasrSTPC21}.
Unfortunately, training machine learning models with DP in the small $\epsilon$ regime usually produces models that perform significantly worse than their non-private counterpart \citep{DBLP:conf/iclr/TramerB21,de2022unlocking}.

Membership inference may not always be the privacy attack that is of most concern. 
For example, in social networks, participation is usually public; recovering privately shared photos or messages from a model trained on social network data is the privacy violation.
These kinds of attacks are referred to as training data reconstruction attacks, and have been successfully demonstrated against a number of machine learning models including language models~\citep{carlini2021extracting, mireshghallah2022quantifying}, generative models~\citep{somepalli2022diffusion,DBLP:journals/corr/abs-2301-13188}, and image classifiers~\citep{balle2022reconstructing,DBLP:journals/corr/abs-2206-07758}.
Recent work~\citep{bhowmick2018protection, balle2022reconstructing, guo2022bounding, guo2022analyzing, stock2022defending} has begun to provide evidence that if one is willing to forgo protection against membership inference, then the $\epsilon$ regime that protects against training data reconstruction is far larger, as predicted by the intuitive reasoning that successful reconstruction requires a significant number of bits about an individual example to be leaked by the model.
This also has the benefit that models trained to protect against training data reconstruction \emph{but not} membership inference do not suffer as large a drop in performance, as less noise is added during training.
Yet, the implications of choosing a concrete $\epsilon$ for a particular application remain unclear since the success of reconstruction attacks can vary greatly depending on the details of the threat model, the strength of the attack, and the criteria of what constitutes a successful reconstruction.

In this paper we re-visit the question of training data reconstruction against image classification models trained with DP-SGD \citep{song2013stochastic,abadi2016deep}, the workhorse of differentially private deep learning.
We choose to concentrate our analysis on DP-SGD because state-of-the-art results are almost exclusively obtained with DP-SGD or other privatized optimizers \citep{de2022unlocking, cattan2022fine, mehta2022large}. 
Our investigation focuses on attacks performed under a strong threat model where the adversary has access to intermediate gradients and knowledge of all the training data except the target of the reconstruction attacks.
This threat model is consistent with the privacy analysis of DP-SGD \citep{abadi2016deep} and the informed adversary implicit in the definition of differential privacy \citep{DBLP:conf/sp/NasrSTPC21,balle2022reconstructing}, and implies that conclusions about the impossibility of attacks in this model will transfer to weaker, more realistic threat models involving real-world attackers.
Our investigation focuses on three main questions.
1) How do variations in the threat model (e.g.\ access to gradients and side knowledge available to the adversary) affect the success of reconstruction attacks?
2) Is it possible to obtain upper bounds on reconstruction success for DP-SGD that match the best known attacks and thus provide actionable insights into how to tune the privacy parameters in practice?
3) Does the standard DP parameter $\epsilon$ provide enough information to characterize vulnerability against reconstruction attacks?

Our contributions are summarized as follows:
\begin{itemize}
    \item We illustrate how changes in the threat model for reconstruction attacks against image classification models can significantly influence their success by comparing attacks with access to the final model parameters, access to intermediate gradients, and access to prior information.
    \item We obtain a tight upper bound on the success of any reconstruction attack against DP-SGD with access to intermediate gradients and prior information. Tightness is shown by providing an attack whose reconstruction success closely matches our bound's predictions.
    \item We provide evidence that the DP parameter $\epsilon$ is not sufficient to capture the success of reconstruction attacks on DP-SGD by showing that different configurations of DP-SGD's hyperparameters leading to the same DP guarantee lead to different rates of reconstruction success.
\end{itemize}

\section{Training Data Reconstruction in DP-SGD}
\label{sec: background}

We start by introducing DP-SGD, the algorithm we study throughout the paper, and then discuss reconstruction attacks with access to either only the final trained model, or all intermediate gradients.
Then we empirically compare both attacks and show that gradient-based attacks are more powerful than model-based attacks, and we identify a significant gap between the success of the best attack and a known lower bound.
As a result, we propose the problem of closing the gap between theoretical bounds and empirical reconstruction attacks.

\paragraph{Differential privacy and DP-SGD.}
Differential privacy \citep{DBLP:conf/tcc/DworkMNS06} formalizes the idea that data analysis algorithms whose output does not overly depend on any individual input data point can provide reasonable privacy protections.
Formally, we say that a randomized mechanism $M$ satisfies $(\epsilon, \delta)$-differential privacy (DP) if, for any two datasets $D, D'$ that differ by one point, and any subset $S$ of the output space we have $ P[M(D)\in S] \leq e^{\epsilon}P[M(D')\in S] + \delta$.
% \begin{align}
%     P[M(D)\in S] \leq e^{\epsilon}P[M(D')\in S] + \delta \enspace.
%     \label{eqn:dp}
% \end{align}
Informally, this means that DP mechanisms bound evidence an adversary can collect (after observing the output) about whether the point where $D$ and $D'$ differ was used in the analysis.
For the bound to provide a meaningful protection against an adversary interested in this membership question it is necessary to take $\epsilon \approx 1$ and $\delta \leq \nicefrac{1}{|D|}$.

A differentially private version of stochastic gradient descent useful for training ML models can be obtained by bounding the influence of any individual sample in the trained model and masking it with noise.
The resulting algorithm is called DP-SGD \citep{abadi2016deep} and proceeds by iteratively updating parameters with a privatized gradient descent step. 
Given a sampling probability $q$, current model parameters $\theta$ and a loss function $\ell(\theta, \cdot)$, the privatized gradient $g$ is obtained by first creating a mini-batch $B$ including each point in the training dataset with probability $q$, summing the $L_2$-clipped gradients\footnote{The clipping operation is defined as $\mathrm{clip}_C(v) = v / \max\{1, \norm{v}/C\}$.} for each point in $B$, and adding Gaussian noise with standard deviation $\sigma \cdot C$ to all coordinates of the gradient: $g \gets 
    \sum_{z \in B}  \mathrm{clip}_{C}\left(\nabla_{\theta}\ell(\theta, z)\right) + \mathcal{N}(0, C^2 \sigma^2 I)$.
% \begin{equation*}
%     \label{eqn:dpsgd}
%     g \gets 
%     \sum_{z \in B}  \mathrm{clip}_{C}\left(\nabla_{\theta}\ell(\theta, z)\right) + \mathcal{N}(0, C^2 \sigma^2 I) \enspace.
% \end{equation*}
Running DP-SGD for $T$ training steps yields an $(\epsilon, \delta)$-DP mechanism with $\epsilon \approx \frac{q \sqrt{T \log(1/\delta)}}{\sigma}$ \citep{abadi2016deep} -- in practice, tighter numerical bounds on $\epsilon$ are often used \citep{NEURIPS2021_6097d8f3,googledp}.
For analytic purposes it is often useful to consider alternatives to $(\epsilon, \delta)$-DP to capture the differences between distributions $M(D)$ and $M(D')$. R{\'e}nyi differential privacy (RDP) \citep{mironov2017renyi} is one such alternative often used in the context of DP-SGD.
It states that the mechanism is $(\alpha, \epsilon)$-RDP for $\alpha > 1$ and $\epsilon \geq 0$ if $\Ex_{W \sim M(D')}\left[\left(\frac{\Pr[M(D) = W]}{\Pr[M(D') = W]}\right)^\alpha \right] \leq e^{(\alpha-1) \epsilon}$.
% \begin{align*}
%     \Ex_{W \sim M(D')}\left[\left(\frac{\Pr[M(D) = W]}{\Pr[M(D') = W]}\right)^\alpha \right] \leq e^{(\alpha-1) \epsilon} \enspace.
% \end{align*}
In particular, $T$ iterations of full-batch DP-SGD (i.e.\ $q=1$) with noise multiplier $\sigma$ give $(\alpha, \frac{\alpha T}{2 \sigma^2})$-RDP for every $\alpha > 1$.

Ultimately, we are interested in understanding how the privacy guarantee of DP-SGD affects \emph{reconstruction attacks} and, in particular, whether $\epsilon \gg 1$ still provides some protection against these more ambitious attacks.
The first step is to understand what the most idoneous threat model to investigating this question is, and then to instantiate a powerful attack in that model.

\subsection{Comparing reconstruction attacks under intermediate and final model access}\label{ssec:comparing_model_and_grad_recon}

In the case of DP-SGD, the privacy guarantee for the final model is obtained by analyzing the privacy loss incurred by \emph{releasing the $T$ private gradients} used by the algorithm. Thus, the guarantee applies both to the intermediate gradients and the final model (by virtue of the post-processing property of DP).
It has been shown through membership inference attacks that the guarantee obtained through this analysis for the \emph{collection of intermediate gradients} can be numerically tight \citep{DBLP:conf/sp/NasrSTPC21}.
However, in some specific settings amenable to mathematical treatment, it has also been shown that the final model produced by DP-SGD can enjoy a stronger DP guarantee than the collection of all intermediate gradients \citep{DBLP:journals/corr/abs-2203-05363,DBLP:journals/corr/abs-2205-13710}.

Although these observations apply the standard membership formulation of DP, they motivate an important question for trying to understand the implications of DP guarantees for reconstruction attacks: \emph{how does access to intermediate gradients affect the success of reconstruction attacks against DP-SGD?}
We investigate this question by introducing and comparing the success of \emph{model-based} and \emph{gradient-based} attacks.
We assume the adversary receives the output of a DP mechanism $M(D)$ and all the input data in $D$ except for one point $z^*$; the adversary's goal is then to produce a reconstruction $\hat{z}$ of the unknown point.
This adversary is referred to as the \emph{informed adversary} in \citet{balle2022reconstructing}.

\paragraph{Model-based training data reconstruction.}
Suppose we consider the output of DP-SGD to be only the final model $\theta = M(D)$.
Under the informed adversary threat model, \citet{balle2022reconstructing} propose a reconstruction attack in which the adversary uses their knowledge of $M$ and $\Dfixed = D \setminus \{z^*\}$ to train a \emph{reconstructor neural network} (RecoNN) capable of mapping model parameters $\theta$ to reconstructed points $\hat{z}$. The training dataset for this network consists of model-point pairs $(\theta^i, z_i)$ where the $z_i$ are auxiliary points representing side knowledge the adversary might possess about the distribution of the target point $z^*$, and the so-called shadow model $\theta^i$ is obtained by applying $M$ to the dataset $\Dfixed \cup \{z_i\}$ obtained by replacing $z^*$ by $z_i$ in the original dataset.
Despite its computational cost, this attack is effective in reconstructing complete training examples from image classification models trained \emph{without} DP on MNIST and CIFAR-10, but fails to produce correct reconstructions when the model is trained with any reasonable setting of DP parameters.

\begin{figure}[t]
\captionsetup{width=\textwidth, justification=raggedright}
  \centering
\begin{subfigure}[t]{0.5\textwidth}
\centering
    \includegraphics[width=\textwidth]{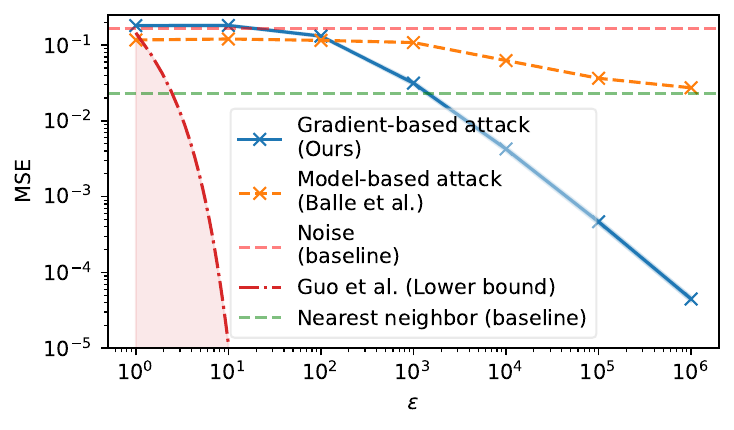}
        \caption{Average MSE between targets and reconstructions on MNIST.}
        \label{fig: experiment_7_mnist_curve}
\end{subfigure}%
\begin{subfigure}[t]{0.5\textwidth}
\centering
    \includegraphics[width=\textwidth]{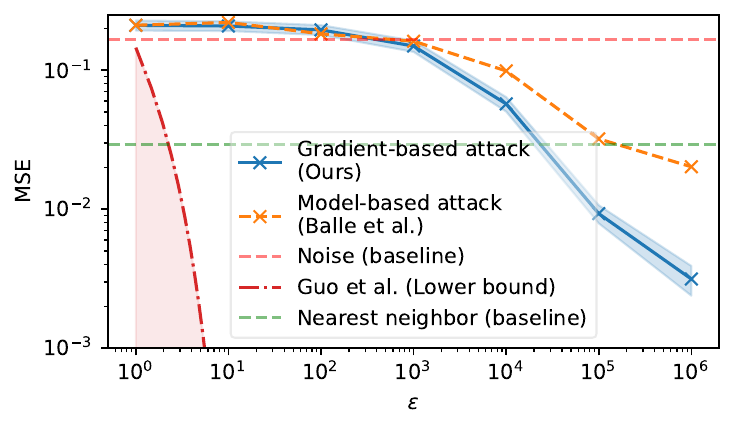}
        \caption{Average MSE between targets and reconstructions on CIFAR-10.}
         \label{fig: experiment_7_cifar_curve}
\end{subfigure}%
\caption{Comparison of model-based and gradient-base reconstruction attacks. We run the attacks over 1K different images for both MNIST and CIFAR-10 for a range of $\epsilon$. In \Cref{fig: experiment_7_mnist_curve} and \Cref{fig: experiment_7_cifar_curve}, we plot the average MSE between reconstruction and target images for both attacks. To help calibrate a reader's interpretation of good and bad reconstructions we plot a nearest neighbor (NN) baseline, marking the average $L_2$-distance between NN images over the entire dataset, and a baseline corresponding to the average distance between random uniform vectors. We also plot a lower bound for MSE given by \citet{guo2022bounding}. We give qualitative examples of reconstructions in \Cref{app: viz_rec_attack}.}
\label{fig: experiment_7}
\end{figure}

\paragraph{Gradient-based training data reconstruction.}
Suppose now that the adversary gets access to all the intermediate gradients $(g_1, \ldots, g_T) = M(D)$ produced by DP-SGD when training model $\theta$.
We can instantiate a reconstruction attack for this scenario by leveraging gradient inversion attacks found in the federated learning literature \citep{yin2021see, huang2021evaluating, jeon2021gradient, jin2021cafe, zhu2019deep,NEURIPS2020_c4ede56b}.
In particular, we use the gradients produced by DP-SGD to construct a loss function $\mathcal{L}(z)$ for which the target point $z^*$ would be optimal if the gradients contained no noise, and then optimize it to obtain our candidate reconstruction.
An important difference between our attack and previous works in the federated learning setting is that we use gradients from multiple steps to perform the attack and only attempt to recover a single point.

More formally, we start by removing from each privatized gradient $g_t$ the gradients of the known points in $D$.
For simplicity, in this section we only consider the full batch case ($q=1$) where every training point in $D$ is included in every gradient step.
Thus, we set $\bar{g}_t = g_t - \sum_{z \in \Dfixed} \mathrm{clip}_C(\nabla_{\theta} \ell(\theta_{t}, z))$,
where $\theta_{t}$ are the model parameters at step $t$.
Note these can be inferred from the gradients and knowledge of the model initialization, which we also assume is given to the adversary.
Similar (but not identical) to \citet{NEURIPS2020_c4ede56b}, we use the loss $\mathcal{L}(z) = \sum_{t=1}^T \mathcal{L}_t(z)$ with
\begin{align}
    \label{eq:grad_opt_loss}
    % \mathcal{L}_t(z) &= 
    % \langle \mathrm{clip}_C(\nabla_{\theta} \ell(\theta_{t}, z)), \bar{g}_t\rangle \\
    % &+ \lVert \mathrm{clip}_C(\nabla_{\theta} \ell(\theta_{t}, z)) - \bar{g}_t \rVert_1 \enspace.\notag
    \mathcal{L}_t(z) = 
    -\langle \mathrm{clip}_C(\nabla_{\theta} \ell(\theta_{t}, z)), \bar{g}_t\rangle 
    + \lVert \mathrm{clip}_C(\nabla_{\theta} \ell(\theta_{t}, z)) - \bar{g}_t \rVert_1 \enspace.
\end{align}
This loss was selected after exhaustively trying all other losses suggested in the aforementioned literature; empirically, \Cref{eq:grad_opt_loss} performed best.

\paragraph{Comparing reconstruction attacks.}
We now compare the success of model-based and gradient-based reconstruction attacks against classification models trained with DP-SGD on MNIST and CIFAR-10. 
We refer to \cref{app: opt_grad_details} for experimental details.

Following \citet{balle2022reconstructing}, we report the mean squared error (MSE) between the reconstruction and target as a function of $\epsilon$ for both attacks. 
Results are shown in \cref{fig: experiment_7}, where we include two baselines to help calibrate how MSE corresponds to good and bad reconstructions.
Firstly, we include a threshold representing the average distance between target points and their nearest neighbor in the remaining of the dataset.
An MSE below this line indicates a near-perfect reconstruction.
Secondly, we include a threshold representing the average distance between random uniform vectors of the same data dimensionality.
An MSE above or close to this line indicates a poor reconstruction.

We make the following observations.
First, the gradient-based attack outperforms the model-based attack by orders of magnitude at larger $\epsilon$ values.
Second, the best attack starts to perform well on MNIST and CIFAR-10 between $10^2<\epsilon<10^3$ and $10^3<\epsilon<10^4$ respectively. This indicates that attack success is affected by the complexity of the underlying data, including its dimension and geometry.
Finally, the attacks give an upper bound for reconstruction error (MSE), however it is not clear if this bound is tight (i.e., whether a better attack could reduce MSE). \citet{guo2022bounding} report a \emph{lower bound} for MSE of the form $\nicefrac{1}{4(e^{\epsilon}-1)}$, which is very far from the upper bound. 
% We note that the concurrent work of \citet{guo2022analyzing} also  provides an upper bound on the success of training data reconstruction under DP. 
% We compare this to our own upper bound in \cref{sec:compare_guo_multi_mia}.

The experiment above illustrates how a change in threat model can make a significant difference in the success of reconstruction attacks. In particular, it shows that the attack from \citet{balle2022reconstructing} is far from optimal, potentially because of its lack of information about intermediate gradients.
Thus, while optimal attacks for membership inference are known \citep{DBLP:conf/sp/NasrSTPC21} and can be used to empirically evaluate the strength of DP guarantees (e.g.\ for auditing purposes \citep{DBLP:journals/corr/abs-2202-12219}), the situation is far less clear in the case of reconstruction attacks. Gradient-based attacks improve over model-based attacks, but \emph{are they (near) optimal?}
Optimality of attacks is important for applications where one would like to calibrate the privacy parameters of DP-SGD to provide a demonstrable, pre-specified amount of protection against any reconstruction attacks.

\section{Bounding Training Data Reconstruction}\label{sec:bounding_reconstruction}

We will now formalize the reconstruction problem further and then provide bounds on the success probability of reconstruction attacks against DP-SGD. 
We will also develop improvements to the reconstruction attack introduced in \Cref{sec: background} from which to benchmark the tightness of our bound.

\subsection{Threat models for reconstruction}

\cite{balle2022reconstructing} introduce a formal definition for reconstruction that attempts to capture the success probability of any reconstruction attack against a given mechanism.
The definition involves an informed adversary with knowledge of the mechanism $M$, the fixed part of the training dataset $\Dfixed$, and a prior distribution $\pi$ from which the target $z^*$ is sampled -- the prior encodes the adversary's side knowledge about the target.

\begin{definition}\label{thm:def_rero}
Let $\pi$ by a prior over target points and $\lattackloss$ a reconstruction error function.
A randomized mechanism $M$ is $(\eta, \gamma)$-ReRo (\emph{reconstruction robust}) with respect to $\pi$ and $\lattackloss$ if for any fixed dataset $\Dfixed$ and any reconstruction attack $R$ we have $\Pr_{Z \sim \pi, w \sim M(\Dfixed \cup \{Z\})}[\lattackloss(Z, \reconstruct(w)) \leq \eta] \leq \gamma$.

\end{definition}

Note that the output of the mechanism does not necessarily need to be final model parameters -- indeed, the definition also applies to attacks operating on intermediate gradients when those are included in the output of the mechanism.
\citet{balle2022reconstructing} also proved that any $(\alpha, \epsilon)$-RDP mechanism is $(\eta, \gamma)$-ReRo with $\gamma = \left(\kappa_{\pi,\lattackloss}(\eta) \cdot e^{\epsilon} \right)^{\frac{\alpha - 1}{\alpha}}$, where $\kappa_{\pi,\lattackloss}(\eta) = \sup_{z_0 \in \Zset} \Pr_{Z \sim \pi}[\lattackloss(Z, z_0) \leq \eta]$.
In particular, using the RDP guarantees of DP-SGD, one obtains that in the full-batch case, running the algorithm for $T$ iterations with noise multiplier $\sigma$ is $(\eta, \gamma)$-ReRo with $\gamma$ bounded by
\begin{align}\label{eqn:dpsgd_rero_balle}
    \exp\left(- \max\left\{0 , \sqrt{\log\frac{1}{\kappa_{\pi,\lattackloss}(\eta)}} - \sqrt{\frac{T}{2 \sigma^2}}\right\}^2 \right) \enspace.
\end{align}
The quantity $\kappa_{\pi,\lattackloss}(\eta)$ can be thought of as the prior probability of reconstruction by an adversary who outputs a fixed candidate $z_0$ based only on their knowledge of the prior (i.e.\ without observing the output of the mechanism).

\paragraph{Instantiating threat models within the $(\eta, \gamma)$-ReRo framework}

In \cref{ssec:comparing_model_and_grad_recon}, we described two variants of a reconstruction attack, one where the adversary has access to intermediate model updates (gradient-based reconstruction) and one where the adversary has access to the final model only (model-based reconstruction).
In the language of $(\eta, \gamma)$-ReRo, both the gradient-based and model-based reconstruction attacks introduced in \Cref{sec: background} take as arguments: $\pi$ -- prior information about the target point, $\Dfixed$ -- the training dataset excluding the target point, $w$ -- the output of the mechanism $M$, and side information about DP-SGD -- such as hyperparameters used in training and how to the model was initialized.
The model-based reconstruction attack assumes that $w=\theta_T$, the parameters of the final model, whereas for gradient-based attacks, $w=(g_1, \ldots, g_T)$, and so the adversary has access to all intermediate privatized model parameter updates.

The gradient-based attack optimizing \Cref{eq:grad_opt_loss} does not make use of any potential prior knowledge the adversary might have about the target point $z^*$, beyond input bounds (the input is bound between [0,1]) and the dimensionality of target ($28\times28$ for an MNIST digit).
% We now present a prior-aware attack whose success can be directly compared to the reconstruction upper bound from \Cref{cor:bound_reconstruction_dpsgd}.
On the other hand, the model-based attack makes use of a prior in the form of the auxiliary points $z_i$ used in the construction of shadow models; these points represent the adversary's knowledge about the distribution from which the target points $z^*$ is sampled.

Going forward in our investigation we will assume the adversary has access to more ``reliable'' side knowledge: their prior is a uniform distribution over a finite set of candidate points $\{z_1, \ldots, z_n\}$, one of which corresponds to the true target point.
This setting is favorable towards the adversary: the points in the prior represent a shortlist of candidates the adversary managed to narrow down using side knowledge about the target.
In this case it is also reasonable to use the error function $\lattackloss(z,z') = \indicator[z \neq z']$ since the adversary's goal is to identify which point from the prior is the target.
As DP assumes the adversary knows (and can influence) all but one record in the training set, the assumption that the adversary has prior knowledge about the target is aligned with the DP threat model.
The main distinction between membership inference and reconstruction with a uniform prior is that in the former the adversary (implicitly) managed to narrow down the target point to two choices, while in the latter they managed to narrow down the target to $n$ choices.
This enables us to smoothly interpolate between the exceedingly strong DP threat model of membership inference (where the goal is to infer a single bit) and a relaxed setting where the adversary's side knowledge is less refined: here $n$ controls the amount of prior knowledge the adversary is privy to, and requires the adversary to infer $\log(n)$ bits to achieve a successful reconstruction.

The discrete prior setup, we argue, provides a better alignment between the success of resconstruction attacks and what actually constitutes privacy leakage. In particular, it allows us to move away from (approximate) verbatim reconstruction as modelled, e.g., by an $\ell_2$ reconstruction criteria success, and model more interesting situations. For example, if the target image contains a car, an attacker might be interested in the digits of the license plate, not the pixels of the image of the car and its background. Thus, if the license plate contains 4 digits, the attacker's goal is to determine which of the possible 10,000 combinations was present in the in the training image.

In \Cref{app:large_prior_exp}, we also conduct experiments where the adversary has less background knowledge about the target point, and so the prior probability of reconstruction is \emph{extremely} small (e.g. $\mathcal{O}\big((\nicefrac{1}{256})^{32\times32\times3}\big)$ for a CIFAR-10 image).

% \borja{It might be helpful to move all the discussion about choice of prior and why uniform prior makes sense to here. WDYT?}
% \jamie{map the different reconstruction attacks to ReRO. Talk about threat model and why it needs to be pessimistic}

\subsection{Reconstruction Robustness of DP-SGD}

\paragraph{ReRo bounds from blow-up functions}
We now state a novel reconstruction robustness bound that, instead of using the (R)DP guarantees of the mechanism, is directly expressed in terms of its output distributions.
The bound depends on the \emph{blow-up function} between two distributions $\mu$ and $\nu$:
\begin{align}\label{eqn:blowup}
    \blowup_{\kappa}(\mu, \nu) = \sup\{\Pr_{\mu}[E] : E \;\; \text{s.t.} \;\; \Pr_{\nu}[E] \leq \kappa \} \enspace.
\end{align}
In particular, let $\nu_{\Dfixed} = M(\Dfixed)$ denote the output distribution of the mechanism on a fixed dataset $\Dfixed$, and $\mu_{\Dz} = M(\Dfixed \cup \{z\})$ for any potential target $z$.
Then we have the following (see \Cref{app:proof_bound_reconstruction} for the full statement).

\begin{theorem}[Informal] \label{thm:bound_reconstruction}
Fix $\pi$ and $\lattackloss$.
Suppose that for every fixed dataset $\Dfixed$ there exists a pair of distributions $\mu^*_{\Dfixed}, \nu^*_{\Dfixed}$ such that $\sup_{z \in \supp(\pi)} \blowup_{\kappa}(\mu_{\Dz}, \nu_{\Dfixed}) \leq \blowup_{\kappa}(\mu^*_{\Dfixed}, \nu^*_{\Dfixed})$ for all $\kappa \in [0,1]$.
Then $M$ is $(\eta, \gamma)$-ReRo with $\gamma = \sup_{\Dfixed} \blowup_{\kappa_{\pi,\lattackloss}(\eta)}(\mu^*_{\Dfixed}, \nu^*_{\Dfixed})$.
\end{theorem}

This result basically says that the probability of successful reconstruction for an adversary that does observe the output of the mechanism can be bounded by the maximum probability under $\mu^*_{\Dfixed}$ over all events that have probability $\kappa_{\pi,\lattackloss}(\eta)$ under $\nu^*_{\Dfixed}$, when we take the worst-case setting over all fixed datasets $\Dfixed$ and all target points $z$ in the support of the prior.
If $M$ satisfies $(\epsilon, 0)$-DP (i.e.\ $(\infty, \epsilon)$-RDP), then $\Pr_{\nu_{\Dfixed}}[E] \leq e^{\epsilon} \Pr_{\mu_{\Dz}}[E]$ for any event $E$, in which case \Cref{thm:bound_reconstruction} gives with $\gamma = \kappa_{\pi,\lattackloss}(\eta) e^{\epsilon}$. This recovers the case $\alpha = \infty$ of the bound from \citet{balle2022reconstructing} stated above.

\begin{remark}The blow-up function is tightly related to the notion of trade-off function defined in \cite{dong2019gaussian}. Precisely,
the trade-off function between two probability distributions $\mu$ and $\nu$ is defined as $\mathcal{T}_{\kappa}(\mu, \nu)=\inf\{\Pr_{\mu}[E] : E \;\; \text{s.t.} \;\; \Pr_{\nu}[E] \leq \kappa \}$. The trade-off function is usually introduced in terms of Type I and Type II errors; it defines the smallest Type II error achievable given constraint on the maximum Type I error, and we have the following relationship: $\mathcal{T}_{\kappa}(\mu, \nu)=1-\blowup_{\kappa}(\mu, \nu)$.
% comment on trade-off an
\end{remark}

\paragraph{ReRo for DP-SGD}
Next we apply this bound directly to DP-SGD without an intermediate appeal to its RDP guarantees; in \Cref{ssec:compare_rero} we will empirically demonstrate that in practice this new bound is vastly superior to the bound in \citet{balle2022reconstructing}.

Let $M$ be the mechanism that outputs all the intermediate gradients produced by running DP-SGD for $T$ steps with noise multiplier $\sigma$ and sampling rate $q$.
Our analysis relies on the observation from \citet{DBLP:journals/corr/abs-2204-06106} that comparing the output distributions of $M$ on two neighbouring datasets can be reduced to comparing the $T$-dimensional Gaussian distribution $\nu_{T,\sigma} = \mathcal{N}(0, \sigma^2 I)$ with the Gaussian mixture $\mu_{T,\sigma,q} = \sum_{w \in \{0,1\}^T}  \Pr[B(q,T) = w] \mathcal{N}(w, \sigma^2 I)$, where $B(q,T)$ is a random binary vector obtained by sampling each coordinate independently from a Bernoulli distribution with success probability $q$.

\begin{corollary}\label{cor:bound_reconstruction_dpsgd}
For any prior $\pi$ and error function $\lattackloss$, $M$ is $(\eta, \gamma)$-ReRo with
$\gamma = \blowup_{\kappa_{\pi,\lattackloss}(\eta)}(\mu_{T,\sigma,q}, \nu_{T,\sigma})$.
\end{corollary}

We defer proofs and discussion to \Cref{app:proofs}.
We note that the concurrent work of \citet{guo2022analyzing} also  provides an upper bound on the success of training data reconstruction under DP. 
We compare this to our own upper bound in \cref{sec:compare_guo_multi_mia}.

\subsection{Upper bound estimation via reconstruction robustness}\label{ssec:estimate_rero}

To estimate the bound in \Cref{cor:bound_reconstruction_dpsgd} in practice we first estimate the baseline probability $\kappa = \kappa_{\pi,\lattackloss}(\eta)$ and then the blow-up $\gamma = \blowup_{\kappa}(\mu_{T,\sigma,q}, \nu_{T,\sigma})$.

\paragraph{Estimating $\kappa$.}
All our experiments deal with a uniform discrete prior $\pi = \frac{1}{n} \sum_i \one_{z_i}$ over a finite set of points $\{z_1, \ldots, z_n\}$. In this case, to compute $\kappa$ it suffices to find the $\lattackloss$-ball of radius $\eta$ containing the largest number of points from the prior. If one is further interested in exact reconstruction (i.e.\ $\lattackloss(z,z') = \indicator[z \neq z']$ and $\eta < 1$) then we immediately have $\kappa = 1/n$.
Although this already covers the setting from our experiments, we show in \Cref{app:estimate_kappa} that more generally it is possible to estimate $\kappa$ only using samples from the prior.
We use this to calculate $\kappa$ for some non-uniform settings where we relax the condition to be on a close reconstruction instead of exact reconstruction.

\paragraph{Estimating $\gamma$.}
Note that it is easy to both sample and evaluate the density of $\mu = \mu_{T,\sigma,q}$ and $\nu = \nu_{T,\sigma}$.
Any time these two operations are feasible we can estimate $\blowup_{\kappa}(\mu, \nu)$ using a non-parametric approach to find the event $E$ that maximizes \Cref{eqn:blowup}.
The key idea is to observe that, as long as $\mu / \nu < \infty$,\footnote{This is satisfied by $\mu = \mu_{T,\sigma,q}$ and $\nu = \nu_{T,\sigma}$.} then a change of measure gives $\Pr_{\mu}[E] = \Ex_{W \sim \nu}\left[ \frac{\mu(W)}{\nu(W)} \indicator[W \in E] \right]$.
% \begin{align*}
%     \Pr_{\mu}[E] = \Ex_{W \sim \nu}\left[ \frac{\mu(W)}{\nu(W)} \indicator[W \in E] \right] \enspace.
% \end{align*}
For fixed $E$ this expectation can be approximated using samples from $\nu$.
Furthermore, since we are interested in the event of probability less than $\kappa$ under $\nu$ that maximizes the expectation, we can take $N$ samples from $\nu$ and keep the $\kappa N$ samples that maximize the ratio $\mu / \nu$ to obtain a suitable candidate for the maximum event.
This motivates the Monte-Carlo approximation of $\gamma$ in \Cref{alg:estimate_gamma}.

\begin{proposition}\label{prop:gamma_consistent}
For $\mu = \mu_{T,\sigma,q}$ and $\nu = \nu_{T,\sigma}$ we have $\lim_{N \to \infty} \hat{\gamma} = \blowup_{\kappa}(\mu, \nu)$ almost surely.
\end{proposition}

\begin{minipage}{0.45\textwidth}
\begin{algorithm}[H]
\caption{Estimating $\gamma$}\label{alg:estimate_gamma}
\begin{algorithmic}
	\State \textbf{Input:} Prior probability $\kappa$, number of samples $N$, sampling and density evaluation access to $\nu$ and $\mu$.
    \State \textbf{Output:} Estimate $\hat{\gamma}$ of reconstruction robustness probability.
    \State Sample $w_1, \ldots, w_N$ independently from $\nu$
    \State Calculate the ratios $r_i = \mu(w_i) / \nu(w_i)$
    \State Sort the ratios in decreasing order: $r_1' \geq \cdots \geq r'_N$
    \State Let $N' = \lceil \kappa N \rceil$ and return $\hat{\gamma} = \frac{1}{N'} \sum_{i=1}^{N'} r'_i$
\end{algorithmic}
\end{algorithm}
\end{minipage}
\hfill
\begin{minipage}{0.45\textwidth}
\begin{algorithm}[H]
\caption{Prior-aware attack}\label{alg: prior_aware_attack}
\begin{algorithmic}
	\State \textbf{Input:} Discrete prior $\pi=\{z_1, \ldots, z_n\}$, Model parameters $\{\theta_1, \theta_1, \ldots, \theta_T\}$, Privatized gradients (with known gradients subtracted) $\{\bar{g}_1, \ldots, \bar{g}_{T}\}$\\
	Observations: $\mathcal{O} \gets \{\}$
	\State \textbf{Output:} Reconstruction guess $\hat{z}\in\pi$
	\For{$i \in [1, 2, \ldots, n]$}
    \State $\mathcal{O}[i] \gets \sum_{t=1}^T \langle \clip_C(\nabla_{\theta} \ell(\theta_t, z_i)), \bar{g}_t \rangle$
    \EndFor
    \State $\hat{i} \gets \argmax \mathcal{O}$ \\
	\Return $\hat{z} \gets \pi[\hat{i}]$ 
\end{algorithmic}
\end{algorithm}
\end{minipage}

\paragraph{Estimation guarantees}
Although Proposition 4 is only asymptotic, we can provide confidence intervals around this estimate using two types of inequalities. 
\cref{alg:estimate_gamma} has two points of error, the gap between the empirical quantile and the population quantile, and the error from mean estimation $\hat{\gamma}$.
The first source of error can be bounded using the DKW inequality~\citep{dvoretzky1956asymptotic} which provides a uniform concentration for each quantile. 
In particular, we can show that the error of quantile estimation is at most $2e^{-2n\eta^2}$, where $n$ is the number of samples and $\eta$ is the error in calculation of the quantile. 
Specifically, with a million samples, we can make sure that with probability 0.999 the error of quantile estimation is less than 0.002, and we can make this smaller by increasing the number of samples. We can account for the second source of error with Bennett's inequality, that leverages the bounded variance of the estimate. In all our bounds, we can show that the error of this part is also less than 0.01 with probability 0.999.
We provide an analysis on the computational cost of estimating $\gamma$ in \cref{app: opt_grad_details}.

\subsection{Lower bound estimation via a prior-aware attack}\label{ssec:prior_aware_attack}

We now present a prior-aware attack whose success can be directly compared to the reconstruction upper bound from \Cref{cor:bound_reconstruction_dpsgd}.

In addition to the uniform prior $\pi = \frac{1}{n} \sum_i \one_{z_i}$ from which we assume the target is sampled from, our prior-aware attack has access to the same information as the gradient-based attack from \Cref{sec: background}: all the privatized gradients (and therefore, intermediate models) produced by DP-SGD, and knowledge of the fixed dataset $\Dfixed$ that can be used to remove all the known datapoints from the privatized gradients.
The attack is given in \Cref{alg: prior_aware_attack}.
For simplicity, in this section we present and experiment with the version of the attack corresponding to the full-batch setting in DP-SGD (i.e.\ $q = 1$). We present an extension of the algorithm to the mini-batch setting in \Cref{sec:effect_of_hyps}.

The rationale for the attack is as follows.
Suppose, for simplicity, that all the gradients are clipped so they have norm exactly $C$.
If $z^*$ is the target sampled from the prior, then $\bar{g}_t \sim \clip_C(\nabla_{\theta} \ell(\theta_{t}, z^*)) + \mathcal{N}(0,C^2 \sigma^2 I)$.
Then the inner products $\langle \clip_C(\nabla_{\theta} \ell(\theta_{t}, z_i)), \bar{g}_t \rangle$ follow a distribution $\mathcal{N}(C^2, C^4 \sigma^2)$ if $z_i = z^*$ and $\mathcal{N}(A, C^4 \sigma^2)$ for some $A < C^2$ otherwise (assuming no two gradients are in the exact same direction).
In particular, $A \approx 0$ if the gradients of the different $z_i$ in the prior are mutually orthogonal.
Thus, finding the $z_i$ that maximizes this sum of inner products is likely to produce a good guess for the target point.
This attack gives a lower bound on the probability of successful reconstruction, which we can compare with the upper bound estimate derived in \Cref{ssec:estimate_rero}.

\section{Experimental Evaluation}
\label{sec:experiments}

We now evaluate both our upper bounds for reconstruction success and our empirical privacy attacks (which gives us lower bounds on reconstruction success). We show that our attack has a success probability nearly identical to the bound given by our theory.
We conclude this section by inspecting how different variations on our threat models change both the upper bound and the optimality of our new attack.

\subsection{Attack and Bound Evaluation on CIFAR-10}
\label{ssec:compare_rero}

We now evaluate our upper bound (\Cref{cor:bound_reconstruction_dpsgd}) and prior-aware reconstruction attack (\Cref{alg: prior_aware_attack}) on full-batch DP-SGD, and compare it to the gradient-based (prior-oblivious) attack optimizing \Cref{eq:grad_opt_loss} and the RDP upper bound obtained from previous work (\Cref{eqn:dpsgd_rero_balle}).
To perform our evaluation on CIFAR-10 models with relatively high accuracy despite being trained with DP-SGD, we follow \citet{de2022unlocking} and use DP-SGD to fine-tune the last layer of a WideResNet model \citep{zagoruyko2016wide} pre-trained on ImageNet.
In this section and throughout the paper, we use a prior with uniform support over $10$ randomly selected points unless we specify the contrary.
Further experimental details are deferred to \Cref{app: opt_grad_details}.

The results of our evaluation are presented in \Cref{fig:experiment_6_main}, which reports the upper and lower bounds on the probability of successful reconstruction produced by the four methods at three different values of $(\epsilon, 10^{-5})$-DP.
The first important observation is that our new ReRo bound for DP-SGD is a significant improvement over the RDP-based bound obtained in \citet{balle2022reconstructing}, which gives a trivial bound for $\epsilon \in \{10, 100\}$ in this setting.
On the other hand, we also observe that the prior-aware attack substantially improves over the prior-oblivious gradient-based attack\footnote{To measure the success probability of the gradient-based attack we take the output of the optimization process, find the closest point in the support of the prior, and declare success if that was the actual target. This process is averaged over 10,000 randomly constructed priors and samples from each prior.}, which \Cref{sec: background} already demonstrated is stronger than previous model-based attacks --
a more thorough investigation of how the gradient-based attack compares to the prior-aware attack is presented in \Cref{sec:grad_attack_comparison} for varying sizes of the prior distribution.
\begin{wrapfigure}{r}{0.45\textwidth}
\centering
\captionsetup{width=.43\textwidth, justification=raggedright}
% \[2em]
\begin{minipage}[t]{0.4\textwidth}
  \centering
  \includegraphics[width=1.\textwidth]{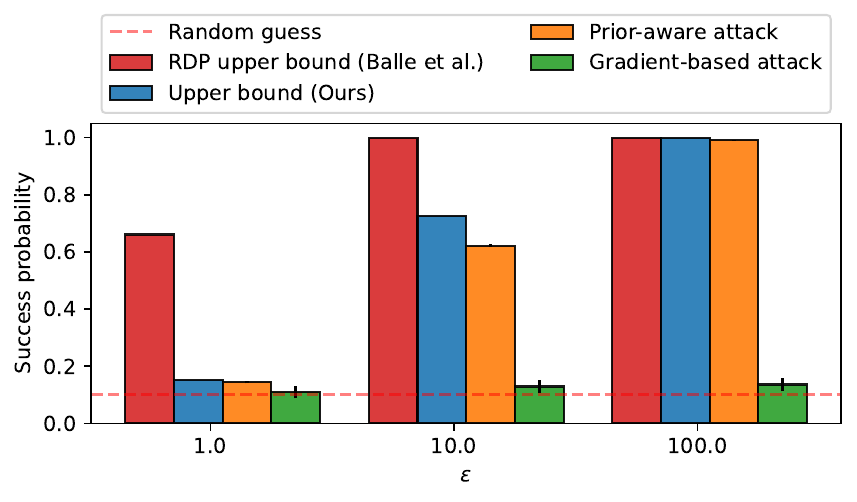}
    \caption{Comparison of success of reconstruction attacks against reconstruction upper bounds on CIFAR-10 with test accuracy at $\epsilon=1, 10$, and $100$ equal to 61.88\%, 89.09\%, and 91.27\%, respectively. Note, the best model we could train at $\epsilon=\infty$ with our experimental settings reaches 93.91\%.}
    \label{fig:experiment_6_main}
\end{minipage}
\\[2em]
\begin{minipage}[t]{0.4\textwidth}
  \centering
\includegraphics[width=1.\textwidth]{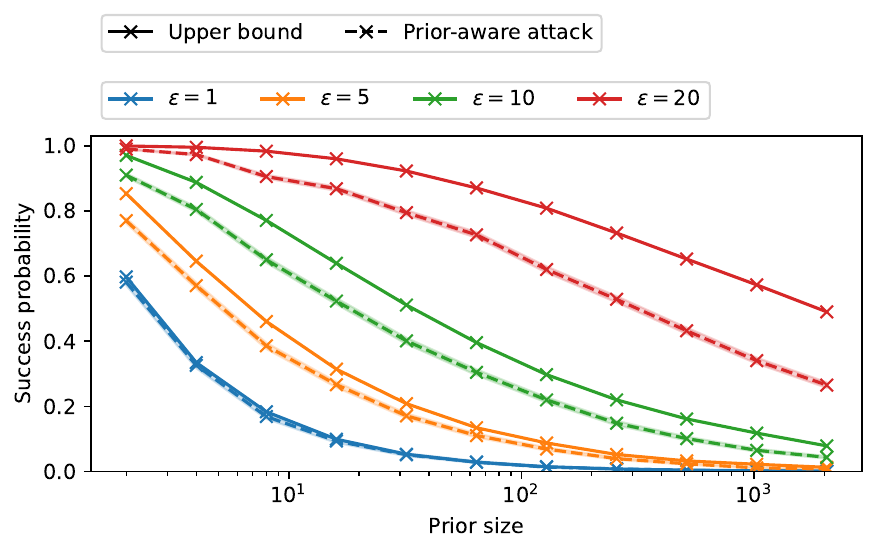}
\caption{How the size of the prior $\pi$ affects reconstruction success probability for a range of different values of $\epsilon$.}
\label{fig: experiment_3}
\end{minipage}
\\[2em]
\begin{minipage}[t]{0.4\textwidth}
    \centering
    \includegraphics[width=1.\textwidth]{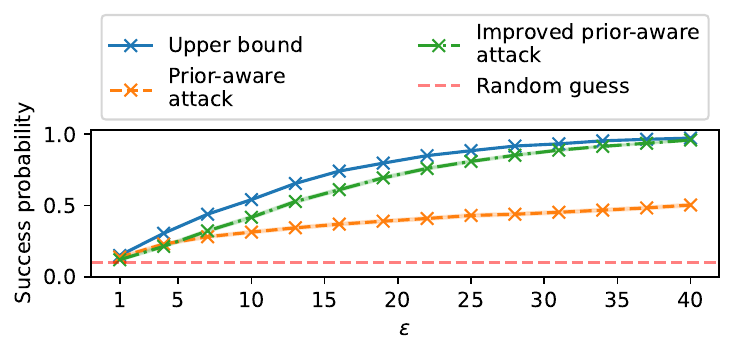}
\caption{We compare how the attack from \Cref{alg: prior_aware_attack} compares with the improved prior-aware attack when $q=0.02$.}
\label{fig: experiment_5a}
\end{minipage}
\vspace{-4em}
\end{wrapfigure}
We will use the prior-aware attack exclusively in following experiments due to its superiority over the gradient-based attack.
Finally, we observe that our best empirical attack and theoretical bound are very close to each other in all settings.
We conclude there are settings where the upper bound is nearly tight and the prior-based attack is nearly optimal.

Now that we have established that our attacks work well on highly accurate large models, our remaining experiments will be performed on MNIST which are significantly more efficient to run.

\subsection{Effects of the prior size}\label{ssec:prior_size}

We now investigate how the bound and prior-aware attack are affected by the size of the prior in \cref{fig: experiment_3}.
As expected, both the upper and lower bound to the probability that we can correctly infer the target point decreases as the prior size increases. 
We also observe a widening gap between the upper bound and attack for larger values of $\epsilon$ and prior size.
However, this bound is still relatively tight for $\epsilon<10$ and a prior size up to $2^{10}$.
Further experiments regarding how the prior affects reconstruction are given in \Cref{app:threat_model_exp_app}.

\subsection{Effect of DP-SGD Hyperparameters}
\label{sec:effect_of_hyps}

Given that we have developed an attack that is close to optimal in a full-batch setting, we are now ready to inspect the effect DP-SGD hyper-parameters controlling its privacy guarantees have on reconstruction success.
First, we observe how the size of mini-batch affects reconstruction success, before measuring how reconstruction success changes with differing hyperparameters at a fixed value of $\epsilon$.
Further experiments that measure the effect of the clipping norm, $C$, are deferred to \Cref{app:dpsgd_hyp_exps_app}.

% \subsection{Effect of sampling rate $q$}\label{sec:subsampling}

\paragraph{Effect of sampling rate $q$.}
All of our experiments so far have been in the full-batch ($q=1$) setting.
We now measure how mini-batching affects both the upper and lower bounds by reducing the data sampling probability $q$ from $1$ to $0.02$ (see \Cref{app: opt_grad_details} for additional experimental details).

The attack presented in \Cref{alg: prior_aware_attack} assumes that the gradient of $z^*$ is present in every privatized gradient.
Since this is no longer true in mini-batch DP-SGD,
we design an improved attack by factoring in the knowledge that only a fraction $q$ of privatized gradients will contain the gradient of the target $z^*$.
This is achieved by, for fixed $z_i$, collecting the inner-products $\langle \clip_C(\nabla_{\theta} \ell(\theta_t, z_i)), \bar{g}_t \rangle$ across all $T$ iterations, and computing the score using only the top $q T$ values.
Pseudo-code for this improved prior-aware attack variant is given in \Cref{app:subsampling_alg}.

Results are shown in \Cref{fig: experiment_5a}, where we observe the improved attack is fairly tight to the upper bound, and also
a large difference between the previous and improved attack.

% \subsection{Effect of DP-SGD hyperparameters at a fixed $\epsilon$}\label{ssec:dsgd_hyp_exp}

\paragraph{Effect of DP-SGD hyperparameters at a fixed $\epsilon$.}

\begin{figure}[t]
\captionsetup{width=1.\textwidth, justification=raggedright}
\begin{subfigure}{0.5\textwidth}
\centering
        \includegraphics[width=\textwidth]{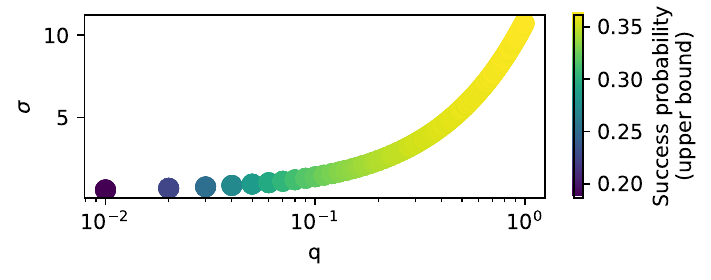}
\caption{Upper bound.}
\label{fig: experiment_same_eps_4_upper}
\end{subfigure}%
\begin{subfigure}{0.5\textwidth}
\centering
    \includegraphics[width=\textwidth]{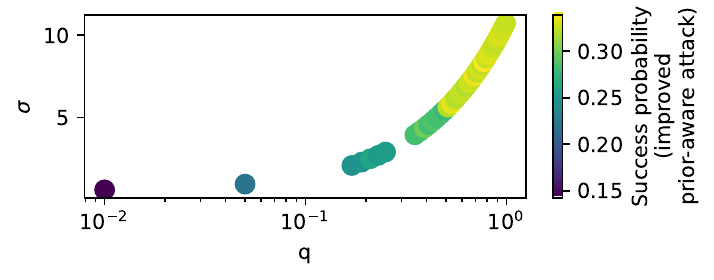}
        \caption{(Improved) Prior-aware attack.}
        \label{fig: experiment_same_eps_4_lower}
\end{subfigure}%

\caption{We plot $\sigma$ against $q$ at $(4, 10^{-5})$-DP at a fixed $T=100$ and the corresponding upper bound (\Cref{fig: experiment_same_eps_4_upper}) and (improved) prior-aware attack (\Cref{fig: experiment_same_eps_4_lower}). Reconstruction is less successful at smaller $q$ in comparison to large $q$ for the same value of $\epsilon$.}
\label{fig: experiment_same_eps_4}
\end{figure}

In our final experiment, we fix $\epsilon$ and investigate if and how the reconstruction upper bound and lower bound (through the improved prior-aware attack) change with different DP-SGD hyperparameters.
In \Cref{fig: experiment_same_eps_4}, we fix $\epsilon=4$ and run DP-SGD for $T=100$ iterations with $C=1$ while varying $q\in [0.01, 0.99]$.
In each setting we tune $\sigma$ to satisfy $(4, 10^{-5})$-DP. 
For experimental efficiency, we use fewer values of $q$ for the attack than the theoretical bound.
Surprisingly, we find that both the upper bound and lower bound \emph{change} with different hyperparameter configurations. 
For example, at $q=0.01$ the upper bound on successful reconstruction is $\approx0.20$ (and the attack is $\approx 0.15$), while at $q=0.99$ it is $\approx0.35$ (and the attack is $\approx 0.32$).
Generally, an increase in $T$ or $q$ increases the upper bound.
This suggests that an $(\epsilon, \delta)$-DP guarantee alone is not able to fully specify the probability that a reconstruction attack is successful.
We conduct a more extensive evaluation for different $\epsilon$ and prior sizes in \Cref{app:fix_epsilon_exp}.

\section{Conclusion}\label{sec: conclusion}
In this work, we investigated training data reconstruction bounds on DP-SGD. 
We developed upper bounds on the success that an adversary can reconstruct a training point $z^*$ under DP-SGD. 
In contrast to prior work that develop reconstruction upper bounds based on DP or RDP analysis of DP-SGD, our upper bounds are directly based on the parameters of the algorithm. 
We also developed new reconstruction attacks, specifically against DP-SGD, that obtain lower bounds on the probability of successful reconstruction, and observe they are close to our upper bounds. 
Our experiments show that both our upper and lower bounds are superior to previously proposed bounds for reconstruction in DP-SGD. 
Our investigations further showed that the $\epsilon$ parameter in DP (or RDP) cannot solely explain robustness to reconstruction attacks.
In other words, one can have the same value of $\epsilon$ for two different algorithms while achieving stronger robustness against reconstruction attacks in one over the other.
This suggest that a more fine-grained analysis of algorithms against specific privacy attacks can lead to superior privacy guarantees and also opens up the possibility of better hyperparameter selection for specific privacy concerns.

\section{Acknowledgements}

The authors would like to thank Leonard Berrada and Taylan Cemgil for their thoughtful feedback on an earlier version of this work.

\bibliography{refs}
\bibliographystyle{icml2023}

\newpage
\onecolumn

\appendix
\section{Experimental details}\label{app: opt_grad_details}

We detail the experimental settings used throughout the paper, and specific hyperparameters used for the various attacks we investigate.
The exact configurations for each experiment are given in \Cref{tab: experiment_details}.
We vary many experimental hyperparameters to investigate their effect on reconstruction, however, the default setting is described next.

For MNIST experiments we use a two layer MLP with hidden width 10 and eLU activations.
The attacks we design in this work perform equally well on all common activation functions, however it is well known that the model-based attack~\citep{balle2022reconstructing} performs poorly on piece-wise linear activations like ReLU.
We set $|\Dfixed|=999$ (and so the training set size is $|\Dfixed \cup \{z^*\}|=1,000$) and train with full-batch DP-SGD for $T=100$ steps.
For each $\epsilon$, we select the learning rate by sweeping over a range of values between 0.001 and 100; we do not use any momentum in optimization.
We set $C=0.1$, $\delta=10^{-5}$ and adjust the noise scale $\sigma$ for a given target $\epsilon$.
The accuracy of this model is over 90\% for $\forall \epsilon \geq 10$, however we emphasize that our experiments on MNIST are meant to primarily investigate the tightness of our reconstruction upper bounds.
We set the size of the prior $\pi$ to ten, meaning the baseline probability of successful reconstruction is 10\%.

For the CIFAR-10 dataset, we use a Wide-ResNet~\citep{zagoruyko2016wide} model with 28 layers and width factor 10 (denoted as WRN-28-10), group normalization, and eLU activations.
We align with the set-up of \citet{de2022unlocking}, who fine-tune a WRN-28-10 model from ImageNet to CIFAR-10.
However, because the model-based attack is highly expensive, we only fine-tune the final layer.
We set $|\Dfixed|=499$ (and so the training set size is $|\Dfixed \cup \{z^*\}|=500$) and train with full-batch DP-SGD for $T=100$ steps; again we sweep over the choice of learning rate for each value of $\epsilon$.
We set $C=1$, $\delta=10^{-5}$ and adjust the noise scale $\sigma$ for a given target $\epsilon$.
The accuracy of this model is over 89\% for $\forall \epsilon \geq 10$, which is close to the state-of-the-art results given by \citet{de2022unlocking}, who achieve 94.2\% with the same fine-tuning setting at $\epsilon=8$ (with a substantially larger training set size).
Again, we set the size of the prior $\pi$ to ten, meaning the baseline probability of successful reconstruction is 10\%.

\begin{table}[t]
\caption{Hyperparameter settings for each experiment. The set of examples that make up the prior ($\pi$), including the target ($z^*$), and the other examples in the training set ($\Dfixed$) are always drawn from the same data distribution, except for the experiment performed in \Cref{fig:experiment_1_main}.}
\label{tab: experiment_details}
\resizebox{\textwidth}{!}{%
\begin{tabular}{llllllll}
\toprule
\multirow{2}{*}{Dataset} & \multirow{2}{*}{Experiment} & \multirow{2}{*}{Clipping norm $C$} & Sampling     & \multirow{2}{*}{Update steps $T$} & \multirow{2}{*}{Model architecture $\theta$}                                                                                                                                                        & Training dataset size                       & \multirow{2}{*}{Prior size}                                   \\
                  &                          &                                    & probability $q$ &                                   &                                                                                                                                                                                                     & ($|\Dfixed| + 1$)                         &                                                               \\
\cmidrule{2-8}
\multirow{2}{*}{CIFAR-10} & \Cref{fig: experiment_7_cifar_curve}, \Cref{fig: experiment_7_cifar_viz}, \Cref{fig: cifar10_linf}                  & 1                 & 1                           & 100              & WRN-28-10                                                                                                                                                                                            & 500                                                                & -                                                             \\
                      &   \Cref{fig:experiment_6_main}                        & 1            & 1                           & 100              & WRN-28-10                                                                                                                                                                                            & 500                                                                  & 10                                                            \\
\cmidrule{2-8}
\multirow{8}{*}{MNIST} & \Cref{fig: experiment_7_mnist_curve}, \Cref{fig: experiment_7_mnist_viz}, \Cref{fig: mnist_linf}                   & 0.1               & 1                           & 100              & MLP $( 784 \rightarrow 10 \rightarrow 10)$                                                                                                                                                        & 1,000                                                                & -                                                             \\
& \Cref{fig: experiment_3}                                                 & 0.1               & 1                           & 100              & MLP $( 784 \rightarrow 10 \rightarrow 10)$                                                                                                                                                        & 1,000                                                               & $2^1 - 2^{11}$ \\
& \Cref{fig: experiment_5a}                                                 & 0.1               & 0.02                        & 1,000            & MLP $( 784 \rightarrow 10 \rightarrow 10)$                                                                                                                                                        & 500                                                                  & 10                                                            \\
& \Cref{fig: experiment_same_eps_4}                                                 &          1      & 0.01-0.99                            & 100               &  MLP $( 784 \rightarrow 10 \rightarrow 10)$ & 1,000                                                                 & 10        \\
& \Cref{fig: experiment_4}                                                 & 0.1               & 1                           & 100              & MLP $( 784 \rightarrow 10, 100, 1000 \rightarrow 10)$ & 1,000                                                               & 10       \\
& \Cref{fig: compare_grad_attack_against_prior_aware_main}                                                 & 0.1               & 1                           & 100              & MLP $( 784 \rightarrow 10 \rightarrow 10)$                                                                                                                                                        & 1,000                                                                & $2^1, 2^3, 2^7$                                               \\
& \Cref{fig:experiment_1_main}                                                 & 0.1               & 1                           & 100              & MLP $( 784 \rightarrow 10 \rightarrow 10)$                                                                                                                                                        & 1,000                                                               & 10                                                            \\
& \Cref{fig: experiment_2_main}                                                & 0.1               & 1                           & 100              & MLP $( 784 \rightarrow 10 \rightarrow 10)$                                                                                                                                                        & 5, 129, 1,000                           & 10                                                            \\
& \Cref{fig: comparison_advantage}                                                 & 1.0              & 1                           & -              & -                                                                                                                                                        &    -                                                           & 10 \\
& \Cref{fig: experiment_5b_main}                                                 & 0.1, 1            & 1                           & 100              & MLP $( 784 \rightarrow 10 \rightarrow 10)$                                                                                                                                                        & 1,000                                                              & 10                                                            \\
& \Cref{fig: compare_dp_hyps_upper_and_lower}                                                 & 1               & 0.01-0.99                           & 100              & MLP $( 784 \rightarrow 10 \rightarrow 10)$                                                                                                                                                        & 1,000                                                                & 2, 10, 100                                               \\
& \Cref{fig: compare_dp_hyps_upper}                                                 & 1               & 0.01-0.99                           & 100              & MLP $( 784 \rightarrow 10 \rightarrow 10)$                                                                                                                                                        & 1,000                                                                & 2, 10, 100                                               \\
\bottomrule
\end{tabular}
}
\end{table}

For the gradient-based and model-based attack we generate 1,000 reconstructions and for prior-aware attack experiments we generate 10,000 reconstructions from which we estimate a lower bound for probability of successful reconstruction.
That is, for experiments in \Cref{sec: background} repeat the attack 1,000 times for targets randomly sampled from base dataset (MNIST or CIFAR-10), and for all other experiments we repeat the attack 10,000 times for targets randomly sampled from the prior, which is itself sampled from the base dataset (MNIST or CIFAR-10).
We now give experimental details specific to the various attacks used throughout the paper.
Note that for attack results, we report 95\% confidence intervals around our lower bound estimate, however, in many cases these intervals are so tight it renders them invisible to the eye.

\paragraph{Model-based attack details.}

For the model-based attack given by \citet{balle2022reconstructing}, we train $40K$ shadow models, and as stated above, construct a test set by training a further 1,000 models on 1,000 different targets (and $\Dfixed$) from which we evaluate our reconstructions.
We use the same architecture for the RecoNN network and optimization hyperparameters as described in the MNIST and CIFAR-10 experiments in \citet{balle2022reconstructing}, and refer the interested reader there for details.

\paragraph{Gradient-based attack details.}

Our optimization hyperparameters are the same for both MNIST and CIFAR-10.
We initialize a $\hat{z}$ from uniform noise and optimize it with respect to the loss given in \Cref{eq:grad_opt_loss} for 1M steps of gradient descent with a learning rate of 0.01.
We found that the loss occasionally diverges and it is useful to have random restarts of the optimization process; we set the number of random restarts to five.
Note we assume that the label of $z^*$ is known to the adversary.
This is a standard assumption in training data reconstruction attacks on federated learning, as \citet{zhao2020idlg} demonstrated the label of the target can be inferred given access to gradients.
If we did not make this assumption, we can run the attack be exhaustively searching over all possible labels. 
For the datasets we consider, this would increase the cost of the attack by a factor of ten.
We evaluate the attack using the same 1,000 targets used to evaluate the model-based attack.

\paragraph{Prior-aware attack details.}

The prior-aware attacks given in \Cref{alg: prior_aware_attack} (and in \Cref{alg: improved_prior_aware_attack}) have no specific hyper-parameters that need to be set.
As stated, the attack proceeds by summing the inner-product defined in \Cref{ssec:prior_aware_attack} over all training steps for each sample in the prior and selecting the sample that maximizes this sum as the reconstruction.
One practical note is that we found it useful to normalize privatized gradients such that the privatized gradient containing the target will be sampled from a Gaussian with unit mean instead of $C^2$, which will be sensitive to choice of $C$ and can lead to numerical precision issues.

\paragraph{Estimating $\gamma$ details.}

As described in \Cref{sec:bounding_reconstruction}, $\nu$ is instantiated as  $\mathcal{N}(0, \sigma^2 I)$, a $T$-dimensional isotropic Gaussian distribution with zero mean, and $\mu$ is given by $\sum_{w \in \{0,1\}^T} p(w) \mathcal{N}(w, \sigma^2 I)$, a mixture of $T$-dimensional isotropic Gaussian distributions with means in $\{0,1\}^T$ sampled according to $B(q, T)$.
Throughout all experiments, we use 1M independent Gaussian samples to compute the estimation of $\gamma$ given by the procedure in \Cref{alg:estimate_gamma}, and because we use a discrete prior of size $|\pi|$, the base probability of reconstruction success, $\kappa$, is given as $\nicefrac{1}{|\pi|}$.
Estimating $\gamma$ is cheap; setting $T=1$ and using a 2.3 GHz 8-Core Intel Core i9 CPU it takes 0.002s to estimate with 10,000 samples. 
This increases to 0.065s, 0.196s, and 2.084s with 100,000, 1M, and 10M samples.
The estimate sensitivity is also small; for the experimental parameters used in \cref{fig:experiment_6_main} at $\epsilon=10$, over 1,000 different calls to \cref{alg:estimate_gamma} the standard deviation of $\gamma$ estimates is $<0.019$ using 10,000 samples, and $<0.0017$ using 1M samples.

\section{Visualization of reconstruction attacks on MNIST and CIFAR-10}\label{app: viz_rec_attack}

In \Cref{fig: experiment_7_viz}, we give a selection of examples for the gradient-based reconstruction attack presented in \Cref{sec: background} and plotted in \Cref{fig: experiment_7}.

\begin{figure*}[t]
\captionsetup{width=\textwidth, justification=raggedright}
  \centering
% \begin{subfigure}[t]{0.5\textwidth}
% \centering
%     \includegraphics[width=\textwidth]{figures/dp_rec_experiment_7_mnist_l2_vs_eps.pdf}
%         \caption{Average MSE between targets and reconstructions on MNIST.}
%         \label{fig: experiment_7_mnist_curve}
% \end{subfigure}%
% \begin{subfigure}[t]{0.5\textwidth}
% \centering
%     \includegraphics[width=\textwidth]{figures/dp_rec_experiment_7_cifar10_l2_vs_eps.pdf}
%         \caption{Average MSE between targets and reconstructions on CIFAR-10.}
%          \label{fig: experiment_7_cifar_curve}
% \end{subfigure}
\begin{subfigure}[t]{0.5\textwidth}
\centering
    \includegraphics[width=1.0\textwidth,height=0.8\textheight, keepaspectratio]{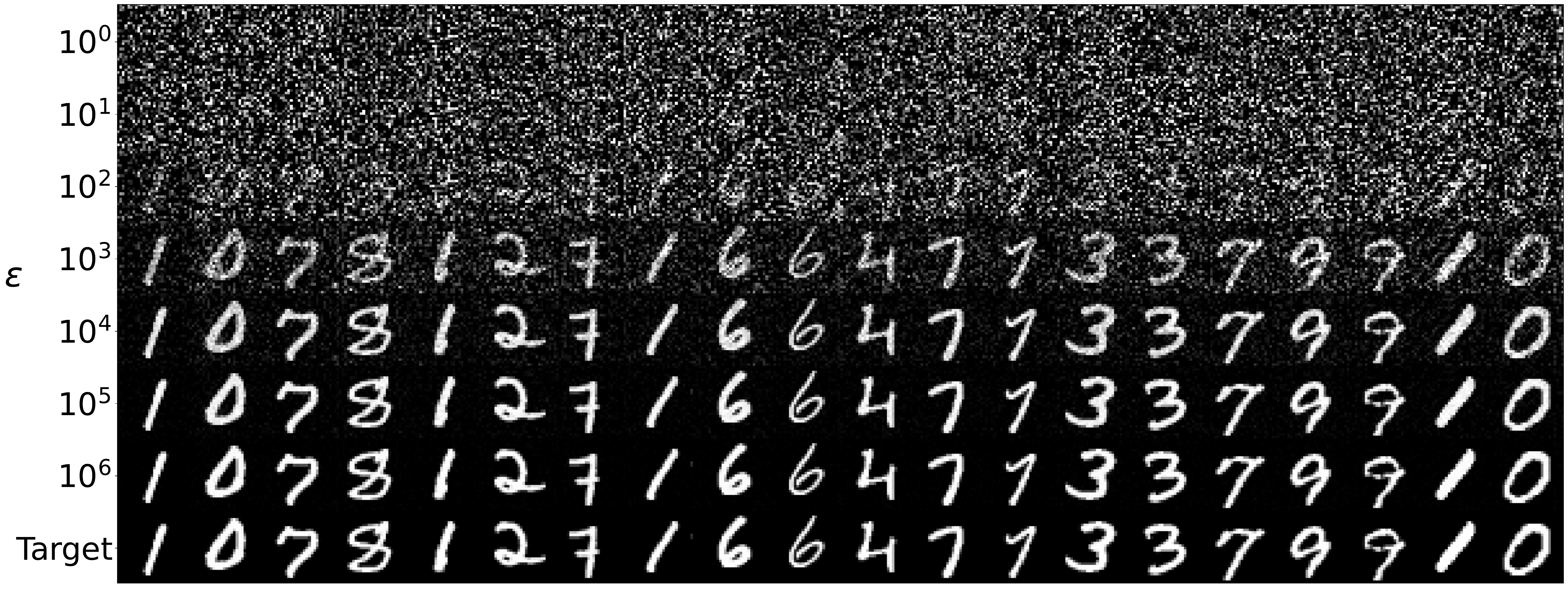}
        \caption{Examples of reconstructions on MNIST.}
        \label{fig: experiment_7_mnist_viz}
\end{subfigure}%
\begin{subfigure}[t]{0.5\textwidth}
\centering
    \includegraphics[width=1.0\textwidth,height=0.8\textheight, keepaspectratio]{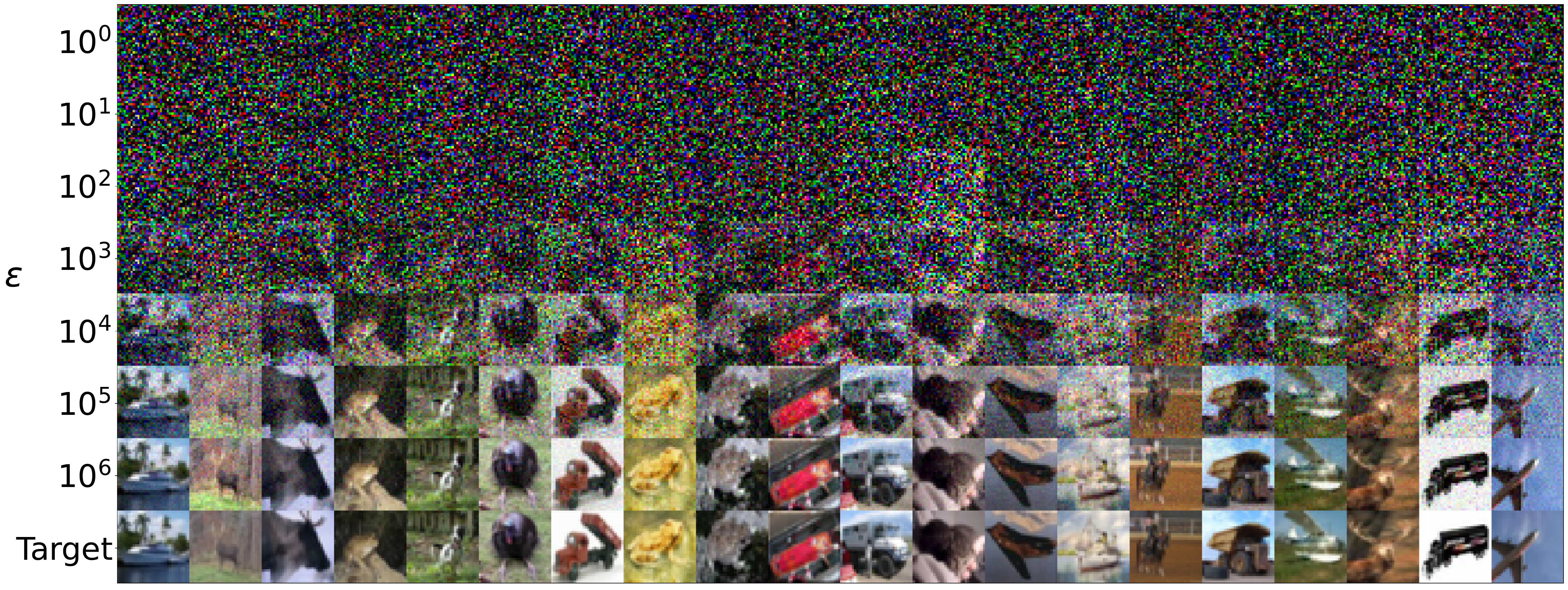}
    \caption{Examples of reconstructions on CIFAR-10.}
    \label{fig: experiment_7_cifar_viz}
\end{subfigure}%
\caption{We give qualitative examples of reconstructions in \Cref{fig: experiment_7_mnist_viz} and \Cref{fig: experiment_7_cifar_viz} for the gradient-based reconstruction attack described in \cref{sec: background}.}
\label{fig: experiment_7_viz}
\end{figure*}

\section{Does the model size make a difference to the prior-aware attack?}

Our results on MNIST and CIFAR-10 suggest that the model size does not impact the tightness of our reconstruction attack (lower bound on probability of success); the MLP model used for MNIST has 7,960 trainable parameters, while the WRN-28-10 model used for CIFAR-10 has 36.5M.
We systematically evaluate the impact of the model size on our prior-aware attack by increasing the size of the MLP hidden layer by factors of ten, creating models with 7,960, 79,600, and 796,000 parameters.
Results are given in \Cref{fig: experiment_4}, where we observe almost no difference in terms of attack success between the different model sizes.

\begin{figure}[t]
\captionsetup{width=1.0\columnwidth, justification=raggedright}
  \centering
\begin{subfigure}{0.33\columnwidth}
\centering
    \includegraphics[width=\columnwidth]{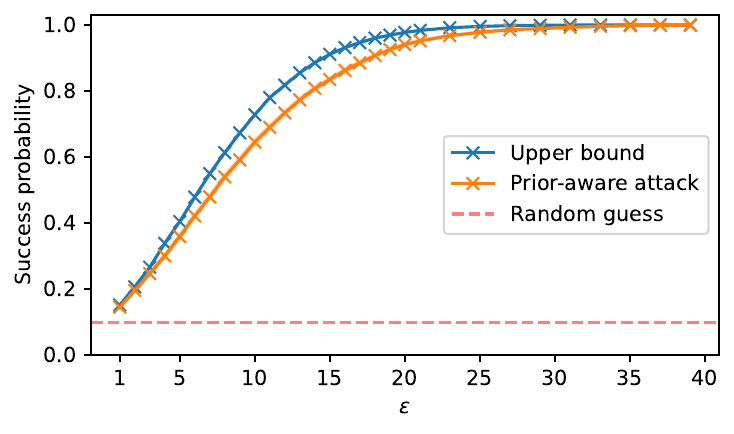}
        \caption{Width: 10.}
\end{subfigure}%
\begin{subfigure}{0.33\columnwidth}
\centering
    \includegraphics[width=\columnwidth]{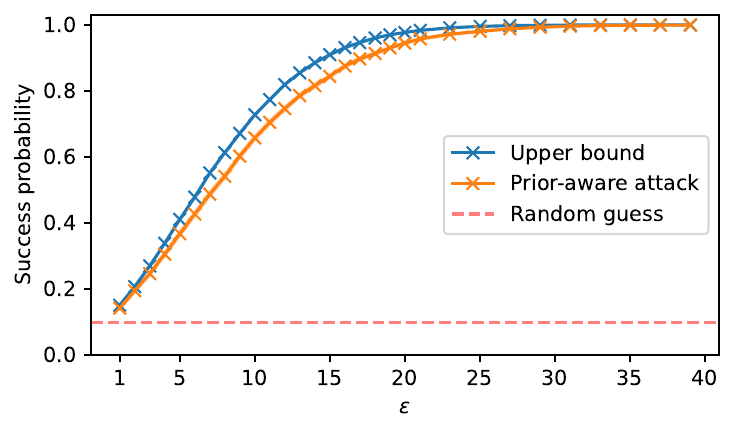}
        \caption{Width: 100.}
\end{subfigure}
\begin{subfigure}{0.33\columnwidth}
\centering
    \includegraphics[width=\columnwidth]{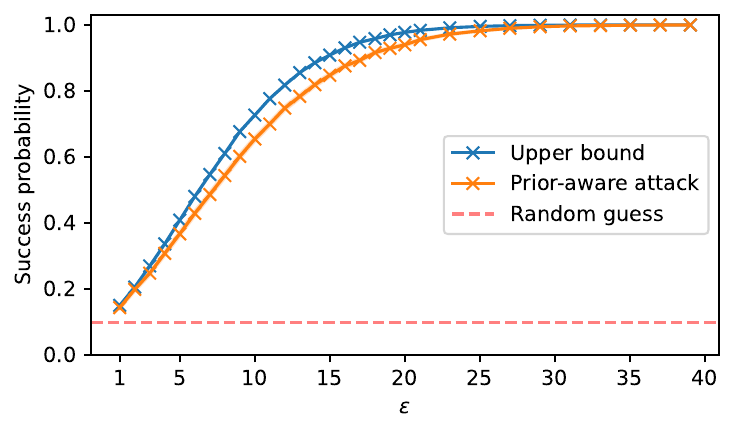}
        \caption{Width: 1000.}
\end{subfigure}%
\caption{Comparison of model sizes on reconstruction by varying the hidden layer width in a two layer MLP.}
\label{fig: experiment_4}
\end{figure}

\section{Comparing the gradient-based attack with the prior-aware attack}\label{sec:grad_attack_comparison}

Our experiments have mainly been concerned with measuring how DP affects an adversary's ability to infer which point was included in training, given that they have access to all possible points that could have been included, in the form of a discrete prior.
This experimental set-up departs from \Cref{fig: experiment_7}, where we assumed the adversary does not have access to a prior set, and so cannot run the prior-aware attack as described in \Cref{alg: prior_aware_attack}.
Following on from results in \Cref{ssec:compare_rero}, we transform these gradient-based attack experimental findings into a probability of successful reconstruction by running a post-processing conversion, allowing us to measure how the assumption of adversarial access to the discrete prior affects reconstruction success.
We run the post-processing conversion in the following way: Given a target sample $z^*$ and a reconstruction $\hat{z}$ found through optimizing the gradient based loss in \cref{eq:grad_opt_loss}, we construct a prior consisting of $z^*$ and $n-1$ randomly selected points from the MNIST dataset, where $n=10$.
We then measure the $L_2$ distance between $\hat{z}$ and every point in this constructed prior, and assign reconstruction a success if the smallest distance is with respect to $z^*$.
For each target $z^*$, we repeat this procedure 1,000 times, with different random selections of size $n-1$, and overall report the average reconstruction success over 1,000 different targets.

This allows us to compare the gradient-based attack (which is prior ``unaware'') directly to our prior-aware attack.
Results are shown in \Cref{fig: compare_grad_attack_against_prior_aware_main}, where we vary the size of the prior between 2, 8, and 128.
In all cases, we see an order of magnitude difference between the gradient-based and prior-aware attack in terms of reconstruction success.
This suggests that if we assume the adversary does not have prior knowledge of the possible set of target points, the minimum value of $\epsilon$ necessary to protect against reconstruction attacks increases.

\begin{figure*}[t]
\captionsetup{width=1.0\textwidth, justification=centering}
  \centering
\begin{subfigure}{0.33\textwidth}
\centering
    \includegraphics[width=\textwidth]{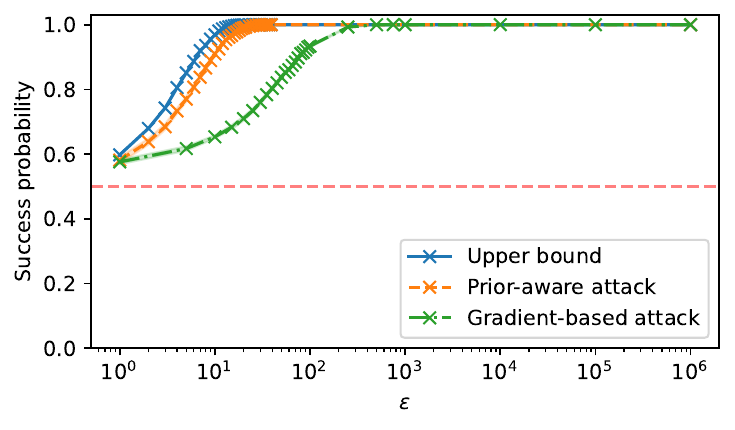}
        \caption{Prior size is 2.}
\end{subfigure}%
\begin{subfigure}{0.33\textwidth}
\centering
    \includegraphics[width=\textwidth]{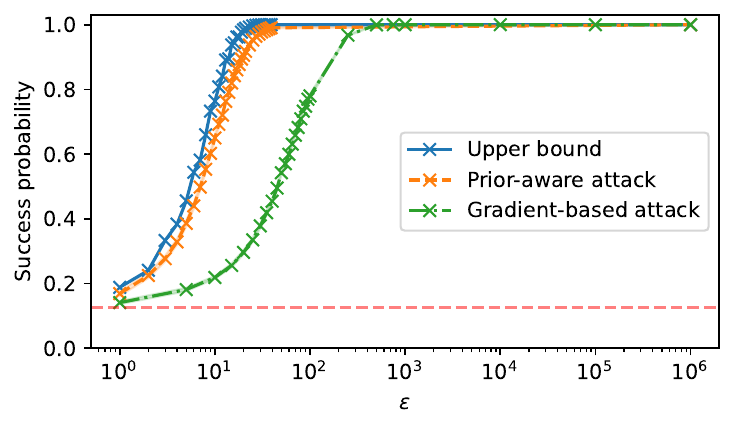}
        \caption{Prior size is 8.}
\end{subfigure}%
\begin{subfigure}{0.33\textwidth}
\centering
    \includegraphics[width=\textwidth]{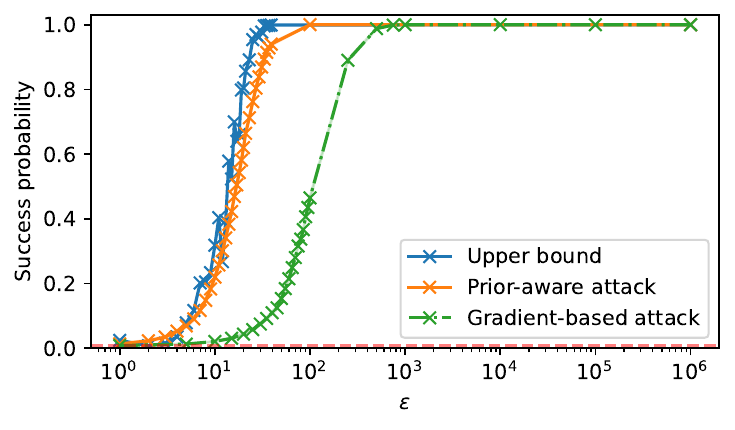}
        \caption{Prior size is 128.}
\end{subfigure}%
\caption{Comparison of prior-aware and gradient-based attack for different prior sizes.}
\label{fig: compare_grad_attack_against_prior_aware_main}
\end{figure*}

\section{Effects of the threat model and prior distribution on reconstruction}\label{app:threat_model_exp_app}

The ability to reconstruct a training data point will naturally depend on the threat model in which the security game is instantiated.
So far, we have limited our investigation to align with the standard adversary assumptions in the DP threat model.
We have also limited ourselves to a setting where the prior is sampled from the same base distribution as $\Dfixed$.
These choices will change the performance of our attack, which is what we measure next.

\paragraph{Prior type.}

We measure how the choice of prior affects reconstruction in \cref{fig:experiment_1_main}.
We train models when the prior is from the same distribution as the rest of the training set (MNIST), and when the prior is sampled random noise.
Note, because the target point $z^*$ is included in the prior, this means we measure how reconstruction success changes when we change the distribution the target was sampled from.
One may expect that the choice of prior to make a difference to reconstruction success if the attack relies on distinguishability between $\Dfixed$ and $z^*$ with respect to some function operating on points and model parameters (e.g. the difference in loss between points in $\Dfixed$ and $z^*$).
However, we see that there is little difference between the two; both are close to the upper bound.

On reflection, this is expected as our objective is simply the sum of samples from a Gaussian, and so the choice of prior may impact our probability of correct inference if this choice affects the probability that a point will be clipped, or if points in the prior have correlated gradients. 
We explore how different values of clipping, $C$, can change reconstruction success probability in \Cref{app:dpsgd_hyp_exps_app}.

\begin{figure}[t]
\captionsetup{width=1.\columnwidth, justification=centering}
    \centering
    \includegraphics[width=0.75\textwidth]{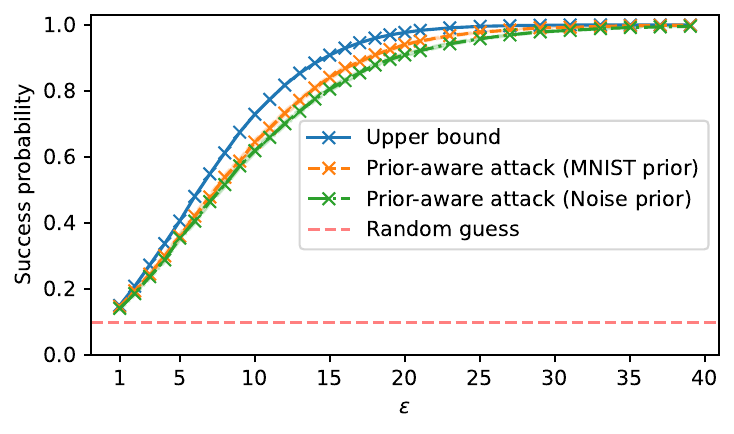}
    \caption{Comparison of how the choice of prior, $\pi$, affects reconstruction success. The prior is selected from a set of examples sampled from MNIST or uniform noise (that has the same intra-sample distance statistics as the MNIST prior).}
    \label{fig:experiment_1_main}
\end{figure}

\paragraph{Knowledge of batch gradients.}

The DP threat model assumes the adversary has knowledge of the gradients of all samples other than the target $z^*$.
Here, we measure how important this assumption is to our attack.
We compare the prior-aware attack (which maximizes $\sum_{t=1}^{T} \langle \clip_C(\nabla_{\theta_t} \ell(z_i)), \bar{g}_t \rangle$) against the attack that selects the $z_i$ maximizing $\sum_{t=1}^{T} \langle \clip_C(\nabla_{\theta_t} \ell(z_i)), g_t \rangle$, where the adversary does not subtract the known gradients from the objective.

In \cref{fig: experiment_2_main}, we compare, in a full-batch setting, when $|\Dfixed|$ is small (set to 4), and see the attack does perform worse when we do not deduct known gradients. 
However, the effect is more pronounced as $|\Dfixed|$ becomes larger, the attack completely fails when setting it to 128.
This is somewhat expected, as with a larger number of samples in a batch it is highly likely there are gradients correlated with the $z^*$ target gradient, masking out its individual contribution and introducing noise into the attack objective.

\begin{figure}[t]
\captionsetup{width=1.\columnwidth, justification=centering}
\centering
    \includegraphics[width=0.75\textwidth]{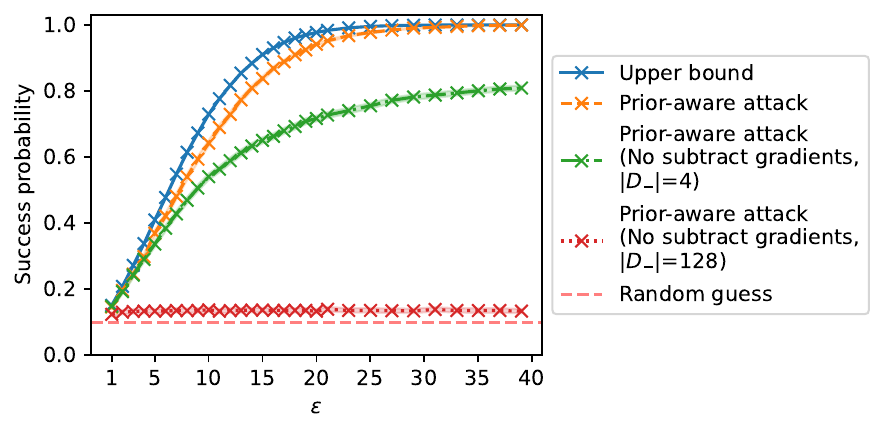}
\caption{In line with the DP threat model, our attack in \Cref{alg: prior_aware_attack} assumes the adversary can subtract known gradients from the privatized gradient. We measure what effect removing this assumption has on reconstruction success probability. When the size of the training set is small, removing this assumption has a minor effect, while reconstruction success drops to random with a larger training set size.}
\label{fig: experiment_2_main}
\end{figure}

\section{Improved prior-aware attack algorithm}\label{app:subsampling_alg}

As explained in \Cref{sec:effect_of_hyps}, the prior-aware attack in \Cref{alg: prior_aware_attack} does not account for the variance introduced into the attack objective in mini-batch DP-SGD, and so we design a more efficient attack specifically for the mini-batch setting.
We give the pseudo-code for this improved prior-aware attack in \Cref{alg: improved_prior_aware_attack}.

\begin{algorithm}[H]
\caption{Improved prior-aware attack}\label{alg: improved_prior_aware_attack}
\begin{algorithmic}
	\State \textbf{Input:} Discrete prior $\pi=\{z_1, \ldots, z_n\}$, Model parameters $\{\theta_1, \theta_1, \ldots, \theta_T\}$, Privatized gradients (with known gradients subtracted) $\{\bar{g}_1, \ldots, \bar{g}_{T}\}$, sampling probability $q$, function that takes the top $qT$ values from a set of observed gradients $top_{qT}$\\
	Observations: $\mathcal{O} \gets \{\}$
	\State \textbf{Output:} Reconstruction guess $\hat{z}\in\pi$
	\For{$i \in [1, 2, \ldots, n]$}
	\State $\mathcal{R} \gets \{\}$
    \For{$t \in [1, 2, \ldots, T]$}
    \State $\mathcal{R}[t] \gets \langle \clip_C(\nabla_{\theta_t} \ell(\theta_t, z_i)), \bar{g}_t \rangle$
    \EndFor
    \State $\mathcal{R} \gets top_{qT}( \mathcal{R} )$ 
    \State $\mathcal{O}[i] \gets sum( \mathcal{R} )$
    \EndFor 
    \State $\hat{i} \gets \argmax \mathcal{O}$ \\
	\Return $\hat{z} \gets \pi[\hat{i}]$ 
\end{algorithmic}
\end{algorithm}

\section{Alternative variant of the prior-aware attack}\label{sec:alt_prior_attack}

Here, we state an alternative attack that uses the log-likelihood to find out which point in the prior set is used for training. Assume we have $T$ steps with clipping threshold $C=1$, noise $\sigma$, and the sampling rate is $q$.

Let $\bar{g}_1, \dots, \bar{g}_T$ be the observed gradients minus the gradient of the examples that are known to be in the batch and let $l_1,\dots,l_T$ be the $\ell_2$ norms of these gradients.

For each example $z$ in the prior set let $g^z_1,\dots, g^z_T$ be the clipped gradient of the example on the intermediate model. Also let $l^z_1,\dots, l^z_T$ be the $\ell_2$ norms of $(\bar{g}_1-g_1^z), \dots, (\bar{g}_T-g_T^z)$.

Now we describe the optimal attack based on $l^z_i$. For each example $z$, calculate the following: $s_z= \sum_{i \in [T]} \ln(1-q + q e^{\frac{-(l^z_i)^2 + l_i^2}{2\sigma^2}})$. It is easy to observe that this is the log probability of outputting the steps conditioned on $z$ being used in the training set. 
Then since the prior is uniform over the prior set, we can choose the $z$ with maximum $s_z$ and report that as the example in the batch.

In fact, this attack could be extended to the non-uniform prior by choosing the example that maximizes $s_z\cdot p_z$, where $p_z$ is the original probability of $z$.

\section{Comparison with \citet{guo2022analyzing}}\label{sec:compare_guo_multi_mia}

% \begin{minipage}{0.45\textwidth}
\begin{figure}[H]
\captionsetup{width=1.0\textwidth, justification=centering}
  \centering
    \includegraphics[width=0.75\textwidth]{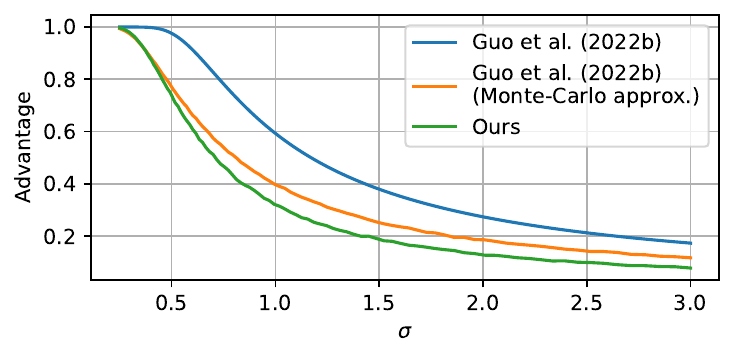}
\caption{Comparison of our upper bound on advantage (\cref{eq:advantage}) with \citet{guo2022analyzing} as function of $\sigma$ for a uniform prior of size ten. We use a single step of DP-SGD with no mini-batch subsampling, and use $100,000$ samples for Monte-Carlo approximation.}
\label{fig: comparison_advantage}
\end{figure}
% \end{minipage}
% \hfill
% \begin{minipage}{0.45\textwidth}

\begin{table}[H]
\centering
\caption{Comparison of our upper bound on advantage (\cref{eq:advantage}) with \citet{guo2022analyzing} and the \citet{guo2022analyzing} Monte-Carlo approximation (abbreviated to MC) as function of $\sigma$ for a uniform prior size of ten and one hundred.}
\label{tab: comparison_guo}
\resizebox{0.75\textwidth}{!}{%
\begin{tabular}{llllllll}
\toprule
Prior size           & Method                                                          & \multicolumn{6}{c}{Advantage upper bound}     \\
\cmidrule{1-8}
\multirow{2}{*}{}    & \multirow{2}{*}{}                                               & \multicolumn{6}{c}{$\sigma$}                  \\
\cmidrule{3-8}
                     &                                                                 & 0.5   & 1     & 1.5   & 2     & 2.5   & 3     \\
\cmidrule{1-8}
\multirow{3}{*}{10}  & \citet{guo2022analyzing}                       & 0.976 & 0.593 & 0.380 & 0.274 & 0.213 & 0.174 \\
                     & \citet{guo2022analyzing} (MC) & 0.771 & 0.397 & 0.257 & 0.184 & 0.144 & 0.118 \\
                     & Ours                                                            & \textbf{0.737} & \textbf{0.322} & \textbf{0.189} & \textbf{0.128} & \textbf{0.099} & \textbf{0.080} \\
\cmidrule{1-8}
\multirow{3}{*}{100} & \citet{guo2022analyzing}                       & 0.861 & 0.346 & 0.195 & 0.131 & 0.097 & 0.076 \\
                     & \citet{guo2022analyzing} (MC) & 0.549 & 0.210  & 0.120  & 0.081 & 0.062 & 0.049 \\
                     & Ours                                                            & \textbf{0.362} & \textbf{0.077} & \textbf{0.035} & \textbf{0.024} & \textbf{0.018} & \textbf{0.012} \\
\bottomrule
\end{tabular}
}
\end{table}
% \end{minipage}

Recently, \citet{guo2022analyzing} have analyzed reconstruction of discrete training data.
They note that DP bounds the mutual information shared between training data and learned parameters, and use Fano's inequality to convert this into a bound on reconstruction success.
In particular, they define the advantage of the adversary as 

\begin{equation}
    \label{eq:advantage}
    \texttt{Adv} := \frac{p_{\text{adversary success}} - p_{\pi}^{\text{max}}}{1 - p_{\pi}^{\text{max}}} \in [0,1].
\end{equation}

where $p_{\pi}^{\text{max}}$ is the maximum sampling probability from the prior, $\pi$, and $p_{\text{adversary success}}$ is the probability that the adversary is successful at inferring which point in the prior was included in training.
They then bound the advantage by lower bounding the adversary's error $t := 1 - p_{\text{adversary success}}$ and by appealing to Fano's inequality they show this can be done by finding the smallest $t \in [0,1]$ satisfying 
\begin{equation}
    \label{eq:fano_problem}
    \begin{split}
    f(t) :=& H(\pi) - I(\pi;w) + t \log t + (1-t) \log(1-t) \\ &- t \log(|\pi|-1)\leq0,
    \end{split}
\end{equation}
where $w$ is output of the private mechanism, $H(\pi)$ is the entropy of the prior, and $I(\pi;w)$ is the mutual information between the prior and output of the private mechanism.
For an $(\alpha, \epsilon)$-RDP mechanism,  $I(\pi;w)\leq\epsilon$, and so $I(\pi;w)$ can be replaced by $\epsilon$ in \Cref{eq:fano_problem}.
However, \citet{guo2022analyzing} show that for the Gaussian mechanism, this can improved upon either by using a Monte-Carlo approximation of $I(\pi;w)$ --- this involves approximating the KL divergence between a Gaussian and a Gaussian mixture --- or by showing that $I(\pi;w)\leq -\sum_{i=1}^{|\pi|} p^i_{\pi} \log\bigg(p^i_{\pi} +(1-p^i_{\pi}) \exp \left( \frac{-\Delta^2}{2\sigma^2} \right)\bigg)$, where $\Delta$ is the sensitivity of the mechanism, and $p^i_{\pi}$ is the probability of selecting the $i$th element from the prior.
We use a uniform prior in all our experiments and so $H(\pi) = -\log(\frac{1}{|\pi|})$ and $p^i_{\pi}=p_{\pi}^{\text{max}}=\frac{1}{|\pi|}$.

We convert our bound on success probability to advantage and compare with the \citet{guo2022analyzing} upper bound (and its Monte-Carlo approximation) in \Cref{fig: comparison_advantage} and \Cref{tab: comparison_guo}, and note our bound is tighter.

\section{Experiments with \emph{very} small priors}\label{app:large_prior_exp}

Our experiments in \Cref{sec:bounding_reconstruction} and \Cref{sec:effect_of_hyps} were conducted with an adversary who has side information about the target point.
Here, we reduce the amount of background knowledge the adversary has about the target, and measure how this affects the reconstruction upper bound and attack success.

We do this in the following set-up:
Given a target $z$, we initialize our reconstruction from uniform noise and optimize with the gradient-based reconstruction attack introduced in \Cref{sec: background} to produce $\hat{z}$.
We mark $\hat{z}$ as a successful reconstruction of $z$ if $\frac{1}{d}\sum_{i=1}^d \indicator[|z[i]-\hat{z}[i]|<\delta]\geq\rho$, where $\rho\in[0,1]$, $d$ is the data dimensionality, and we set $\delta=\frac{32}{255}$ in our experiments.
If $\rho=1$ this means we mark the reconstruction as successful if $\lVert \hat{z} - z\rVert_{\infty} < \delta$, and for $\rho<1$, then at least a fraction $\rho$ values in $\hat{z}$ must be within an $\ell_{\infty}$ ball of radius $\delta$ from $z$. 
Under the assumption the adversary has no background knowledge of the target point, with $\delta=\frac{32}{255}$ and a uniform prior, the prior probability of reconstruction is given by $(\nicefrac{2\times32}{256})^{d\rho}$ --- if $\rho=1$, for MNIST and CIFAR-10, this means the prior probability of a successful reconstruction is $9.66\times 10^{-473}$ and $2.96\times 10^{-1850}$, respectively.

We plot the reconstruction upper bound compared to the attack success for different values of $\rho$ in \Cref{fig: small_prior}.
We also visualize the quality of reconstructions for different values of $\rho$.
Even for $\rho=0.6$, where $40\%$ of the reconstruction pixels can take any value, and the remaining $60\%$ are within an absolute value of $\frac{32}{255}$ from the target, one can easily identify that the reconstructions look visually similar to the target.

\begin{figure}[H]
\captionsetup{width=1.0\columnwidth, justification=centering}
  \centering
\begin{subfigure}{0.5\columnwidth}
\centering
    \includegraphics[width=\columnwidth]{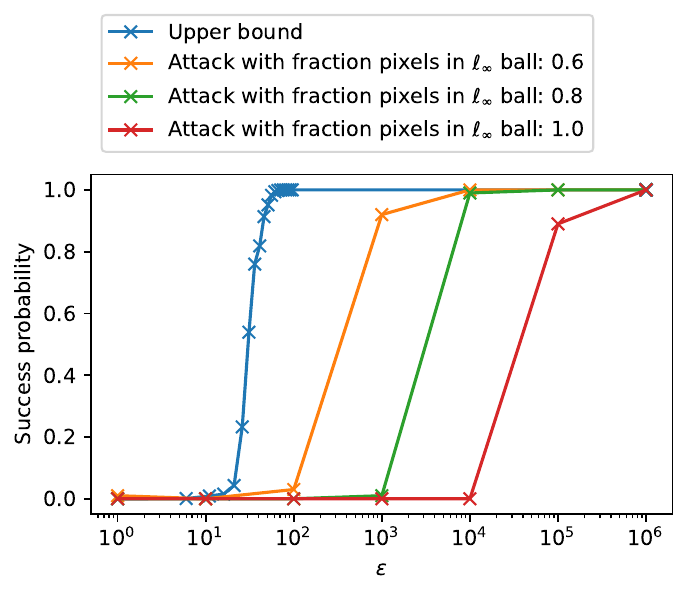}
        \caption{Comparison of reconstruction success under a \emph{very} small prior for MNIST, where we judge a reconstruction as successful if at least $\rho$ pixels are within an absolute distance of $\frac{32}{255}$ of the target.}
        \label{fig: mnist_linf}
\end{subfigure}%
\begin{subfigure}{0.5\columnwidth}
\centering
    \includegraphics[width=\columnwidth]{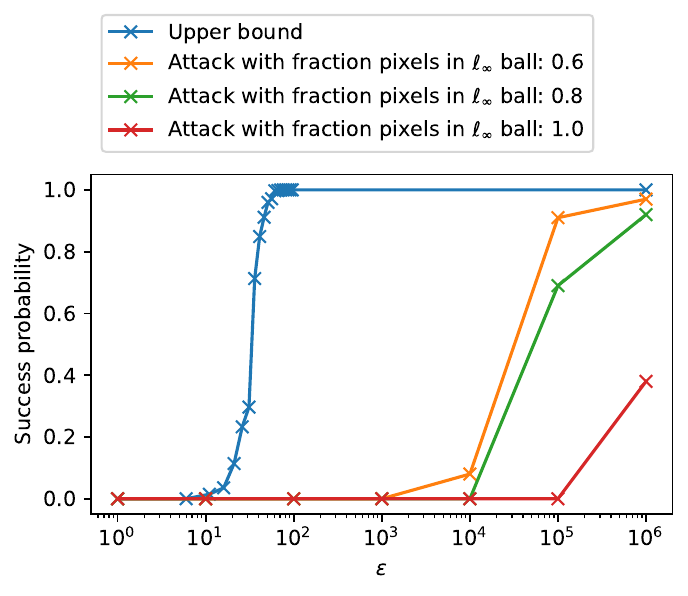}
       \caption{Comparison of reconstruction success under a \emph{very} small prior for CIFAR-10, where we judge a reconstruction as successful if at least $\rho$ pixels are within an absolute distance of $\frac{32}{255}$ of the target.}
       \label{fig: cifar10_linf}
\end{subfigure}

\begin{subfigure}{0.5\columnwidth}
\centering
    \includegraphics[width=\columnwidth]{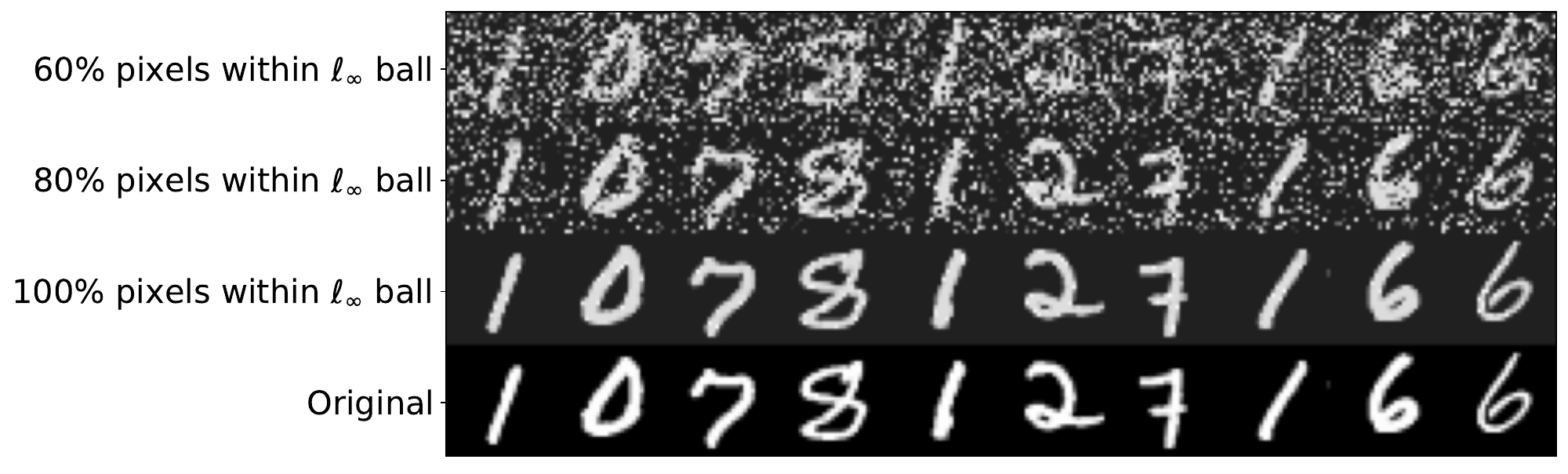}
        \caption{MNIST examples of reconstructions where at least $\rho$ pixels are within an absolute distance of $\frac{32}{255}$ of the target.}
\end{subfigure}%
\begin{subfigure}{0.5\columnwidth}
\centering
    \includegraphics[width=\columnwidth]{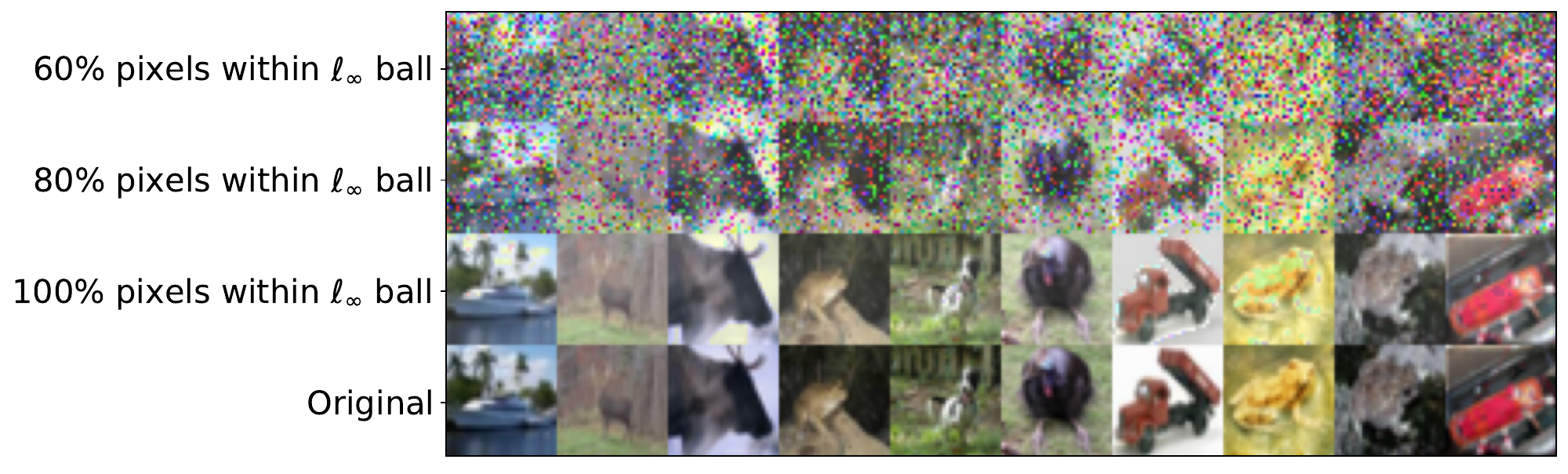}
        \caption{CIFAR-10 examples of reconstructions where at least $\rho$ pixels are within an absolute distance of $\frac{32}{255}$ of the target.}
\end{subfigure}

\caption{Comparison of reconstruction success under a \emph{very} small prior. The prior probability of success for MNIST and CIFAR-10 are $9.66\times 10^{-473}$ and $2.96\times 10^{-1850}$, respectively.}
\label{fig: small_prior}
\end{figure}

\section{Discussion on related work}
\label{sec:relatedwork}

Here, we give a more detailed discussion of relevant related work over what is surfaced in \Cref{sec: introduction} and \cref{sec: background}.

\paragraph{DP and reconstruction.}

By construction, differential privacy bounds the success of a membership inference attack, where the aim is to infer if a point $z$ was in or out of the training set.
While the connection between membership inference and DP is well understood, less is known about the relationship between training data reconstruction attacks and DP.
A number of recent works have begun to remedy this in the context of models trained with DP-SGD by studying the value of $\epsilon$ required to thwart training data reconstruction attacks~\citep{bhowmick2018protection, balle2022reconstructing, guo2022bounding, guo2022analyzing, stock2022defending}.
Of course, because differential privacy bounds membership inference, it will also bound ones ability to reconstruct training data; if one cannot determine if $z$ was used in training, they will not be able to reconstruct that point.
These works are interested in both formalizing training data reconstruction attacks, and quantifying the necessary $\epsilon$ required to bound its success.
Most of these works share a common finding -- the $\epsilon$ value needed for this bound is much larger than the value required to protect against membership inference attacks ($<10$ in practice). 
If all other parameters in  $\nicefrac{q\sqrt{T\log(\frac{1}{\delta})}}{\varepsilon}$ remain fixed, one can see that a larger value of $\epsilon$ reduces the scale of noise we add to gradients, which in turn results in models that achieve smaller generalization error than models trained with DP-SGD that protect against membership inference.

The claim that a protection against membership inference attacks also protects against training data reconstruction attacks glosses over many subtleties.
For example, if $z$ was not included in training it could still have a non-zero probability of reconstruction if samples that are close to $z$ were included in training.
\citet{balle2022reconstructing} take the approach of formalizing training reconstruction attacks in a Bayesian framework, where they compute a prior probability of reconstruction, and then find how much more information an adversary gains by observing the output of DP-SGD.

\citet{balle2022reconstructing} use an average-case definition of reconstruction over the output of a randomized mechanism.
In contrast, \citet{bhowmick2018protection} define a worst-case formalization, asking when should an adversary not be able to reconstruct a point of interest regardless of the output of the mechanism.
Unfortunately, such worst-case guarantees are not attainable under DP-relaxations like $(\epsilon, \delta)$-DP and RDP, because the privacy loss is not bounded; there is a small probability that the privacy loss will be high.

\citet{stock2022defending} focus on bounding reconstruction for language tasks.
They use the probability preservation guarantee from RDP to derive reconstruction bounds, showing that the length of a secret within a piece of text itself provides privacy. 
They translate this to show a smaller amount of DP noise is required to protect longer secrets.

While \citet{balle2022reconstructing} propose a Bayesian formalization for reconstruction error, \citet{guo2022bounding} propose a frequentist definition.
They show that if $M$ is $(2, \epsilon)$-RDP, then the reconstruction MSE is lower bounded by $\nicefrac{\sum_{i=1}^d\text{diam}_i(\Zset)^2}{4d(e^{\epsilon}-1)}$, where $\text{diam}_i(\Zset)$ is the diameter of the space $\Zset$ in the $i$th dimension.

\paragraph{Gradient inversion attacks.}

The early works of \citet{wang2019beyond} and \citet{zhu2019deep} showed that one can invert single image representation from gradients of a deep neural network. 
\citet{zhu2019deep} actually went beyond this and showed one can jointly reconstruct both the image and label representation.
The idea is that given a target point $z$, a loss function $\ell$, and an observed gradient (wrt to model parameters $\theta$) $g_z=\nabla_{\theta}\ell(\theta, z)$, to construct a $\hat{z}$ such that $\hat{z} = \argmin_{z'}\lVert g_{z'} - g_z\rVert$.
The expectation is that images that have similar gradients will be visually similar.
By optimizing the above objective with gradient descent, \citet{zhu2019deep} showed that one can construct visually accurate reconstruction on standard image benchmark datasets like CIFAR-10.

\citet{jeon2021gradient, yin2021see, jin2021cafe, huang2021evaluating, NEURIPS2020_c4ede56b} proposed a number of improvements over the reconstruction algorithm used in \citet{zhu2019deep}: they showed how to reconstruct multiple training points in batched gradient descent, how to optimize against batch normalization statistics, and incorporate priors into the optimization procedure, amongst other improvements.

The aforementioned attacks assumed an adversary has access to gradients through intermediate model updates.
\citet{balle2022reconstructing} instead investigate reconstruction attacks when adversary can only observe a model after it has finished training, and propose attacks against (parametric) ML models under this threat model.
However, the attack they construct is computationally demanding as it involves retraining thousands of models.
This computational bottleneck is also a factor in \citet{DBLP:journals/corr/abs-2206-07758}, who also investigate training data reconstruction attacks where the adversary has access only to final model parameters.

\section{More experiments on the effect of DP-SGD hyperparameters}\label{app:dpsgd_hyp_exps_app}

We extend on our investigation into the effect that DP-SGD hyperparameters have on reconstruction.
We begin by varying the clipping norm parameter, $C$, and measure the effect on reconstruction.
Following this, we replicate our results from \Cref{sec:effect_of_hyps} (the effect hyperparameters have on reconstruction at a fixed $\epsilon$) across different values of $\epsilon$ and prior sizes, $|\pi|$.

\subsection{Effect of clipping norm}

If we look again at our attack set-up in \Cref{alg: prior_aware_attack}, we see that in essence we are either summing a set of samples only from a Gaussian centred at zero or a Gaussian centred at $C^2$. 
If the gradient of the target point is not clipped, then this will reduce the sum of gradients when the target is included in a batch, as the Gaussian will be centred at a value smaller than $C^2$.
This will increase the probability that the objective is not maximized by the target point.

We demonstrate how this changes the reconstruction success probability by training a model for 100 steps with a clipping norm of 0.1 or 1, and measuring the average gradient norm of all samples over all steps.
Results are shown in \cref{fig: experiment_5b_main}.
We see at $C=0.1$, our attack is tight to the upper bound, and the average gradient norm is $0.1$ for all values of $\epsilon$; all individual gradients are clipped.
When $C=1$, the average gradient norm decreases from $0.9$ at $\epsilon=1$ to $0.5$ at $\epsilon=40$, and we see a larger gap between upper and lower bounds. 
The fact that some gradients may not be clipped is not taken into account by our theory used to compute upper bounds, and so we conjecture that the reduction is reconstruction success is a real effect rather than a weakness of our attack.

We note that these findings chime with work on individual privacy accounting~\citep{feldman2021individual, yu2022per, ligett2017accuracy, redberg2021privately}.
An individual sample's privacy loss is often much smaller than what is accounted for by DP bounds.
These works use the gradient norm of an individual sample to measure the true privacy loss, the claim is that if the gradient norm is smaller than the clipping norm, the amount of noise added is too large, as the DP accountant assumes all samples are clipped.
Our experiments support the claim that there is a disparity in privacy loss between samples whose gradients are and are not clipped.

\begin{figure}[t]
\captionsetup{width=1.0\columnwidth, justification=centering}
  \centering
\begin{subfigure}{0.5\columnwidth}
\centering
    \includegraphics[width=\columnwidth]{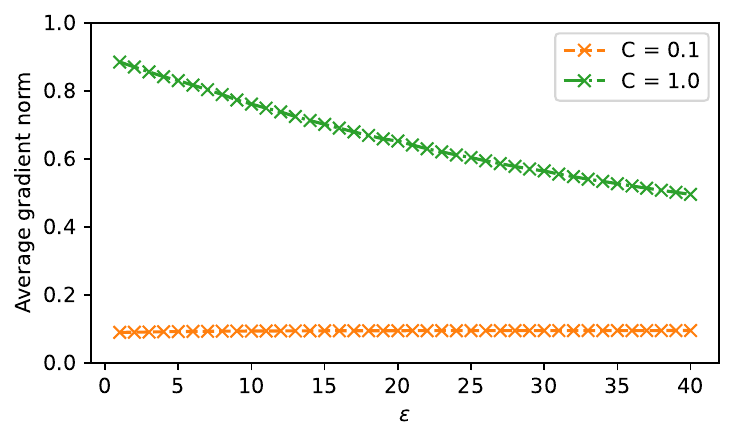}
        \caption{Average gradient norm (over all samples and steps) for different values of $\epsilon$ at $C=0.1$ and $C=1$.}
\end{subfigure}%
\begin{subfigure}{0.5\columnwidth}
\centering
    \includegraphics[width=\columnwidth]{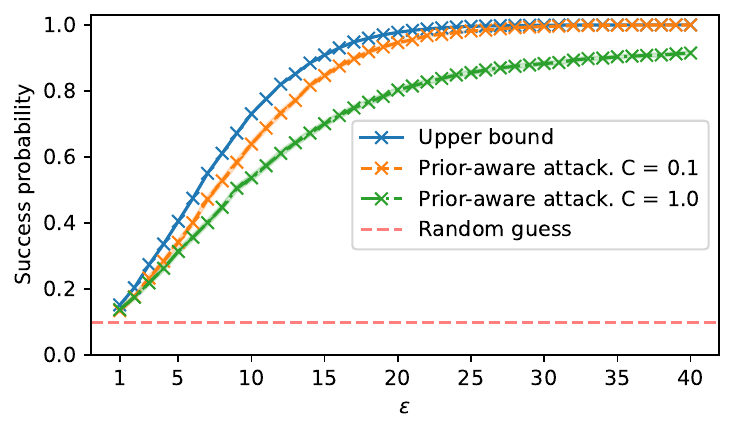}
        \caption{Reconstruction success probability for different values of $\epsilon$ at $C=0.1$ and $C=1$.}
\end{subfigure}%
\caption{Comparison of how reconstruction success is changes with the clipping norm, $C$. We see that if examples have a gradient norm smaller than $C$, and so are not clipped, reconstruction success probability becomes smaller.}
\label{fig: experiment_5b_main}
\end{figure}

\subsection{More results on the effect of DP-SGD hyperparameters at a fixed $\epsilon$}\label{app:fix_epsilon_exp}

In \Cref{sec:effect_of_hyps}, we demonstrated that the success of a reconstruction attack cannot be captured only by the $(\epsilon, \delta)$ guarantee, when $\epsilon=4$ and the size of the prior, $\pi$, is set to ten.
We now observe how these results change across different $\epsilon$ and $|\pi|$, where we again fix the number of updates to $T=100$, $C=1$, vary $q\in[0.01, 0.99]$, and adjust $\sigma$ accordingly.

Firstly, in \Cref{fig: compare_dp_hyps_upper_and_lower}, we measure the upper \emph{and} lower bound ((improved) prior-aware attack) on the probability of successful reconstruction across different $q$. 
In all settings, we observe smaller reconstruction success at smaller $q$, where the largest fluctuations in reconstruction success are for larger values of $\epsilon$.
We visualise this in another way by plotting $\sigma$ against $q$ and report the upper bound in \Cref{fig: compare_dp_hyps_upper}.
Note that the color ranges in \Cref{fig: compare_dp_hyps_upper} are independent across subfigures.

\begin{figure}[t]
\captionsetup{width=1.0\textwidth, justification=centering}
  \centering
\begin{subfigure}{0.25\textwidth}
\centering
    \includegraphics[width=\textwidth]{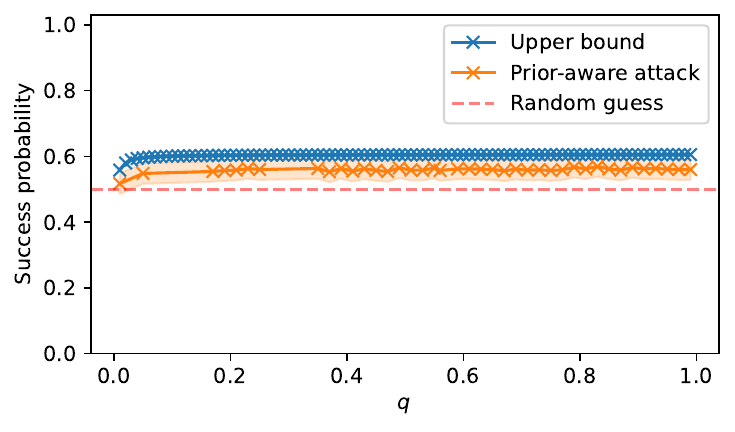}
        \caption{$\epsilon=1, |\pi|=2$.}
\end{subfigure}%
\begin{subfigure}{0.25\textwidth}
\centering
    \includegraphics[width=\textwidth]{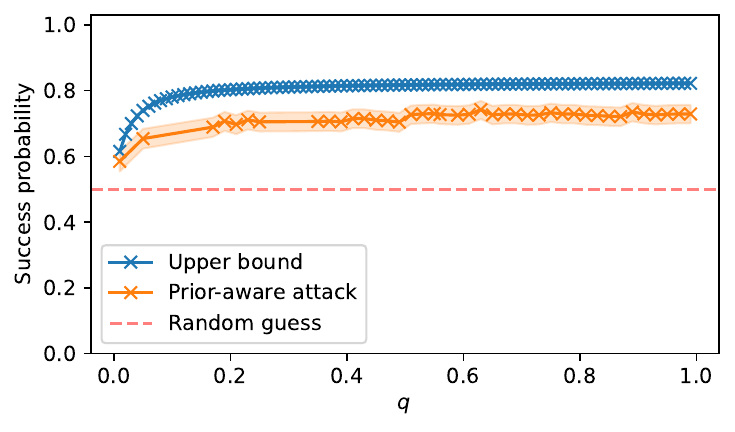}
        \caption{$\epsilon=4, |\pi|=2$.}
\end{subfigure}%
\begin{subfigure}{0.25\textwidth}
\centering
    \includegraphics[width=\textwidth]{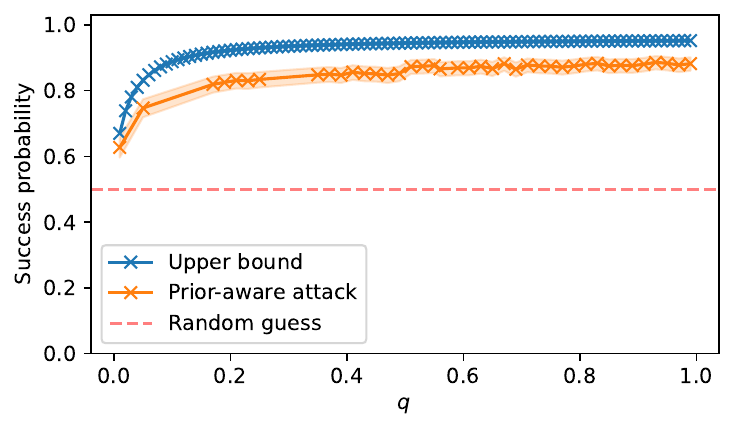}
        \caption{$\epsilon=8, |\pi|=2$.}
\end{subfigure}%
\begin{subfigure}{0.25\textwidth}
\centering
    \includegraphics[width=\textwidth]{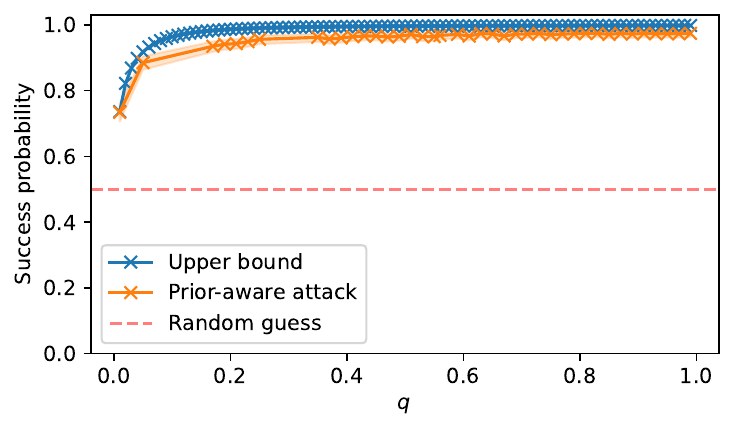}
        \caption{$\epsilon=16, |\pi|=2$.}
\end{subfigure}

\begin{subfigure}{0.25\textwidth}
\centering
    \includegraphics[width=\textwidth]{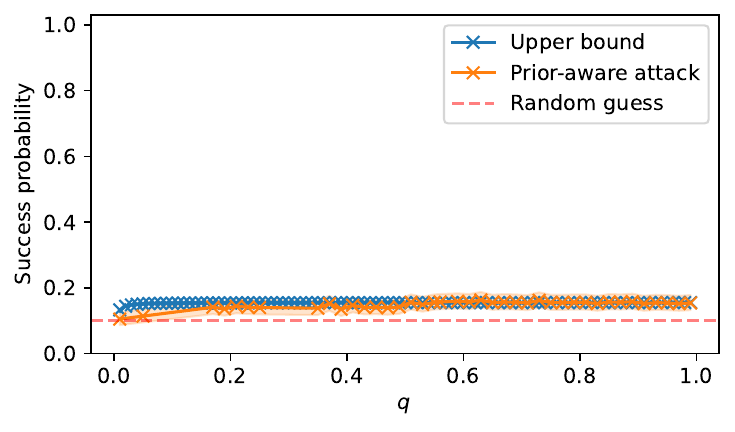}
        \caption{$\epsilon=1, |\pi|=10$.}
\end{subfigure}%
\begin{subfigure}{0.25\textwidth}
\centering
    \includegraphics[width=\textwidth]{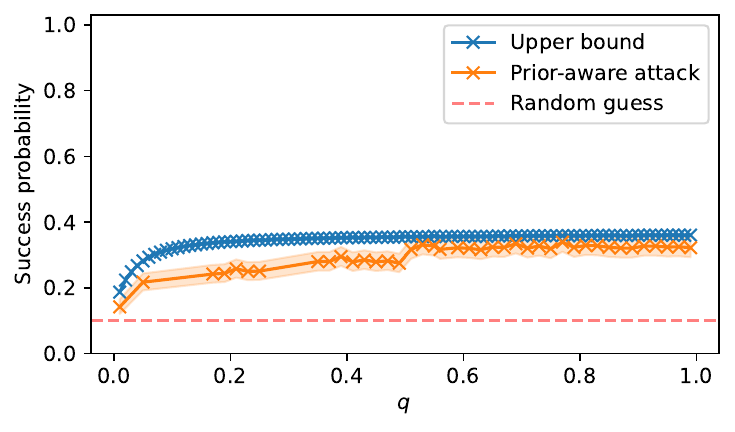}
        \caption{$\epsilon=4, |\pi|=10$.}
\end{subfigure}%
\begin{subfigure}{0.25\textwidth}
\centering
    \includegraphics[width=\textwidth]{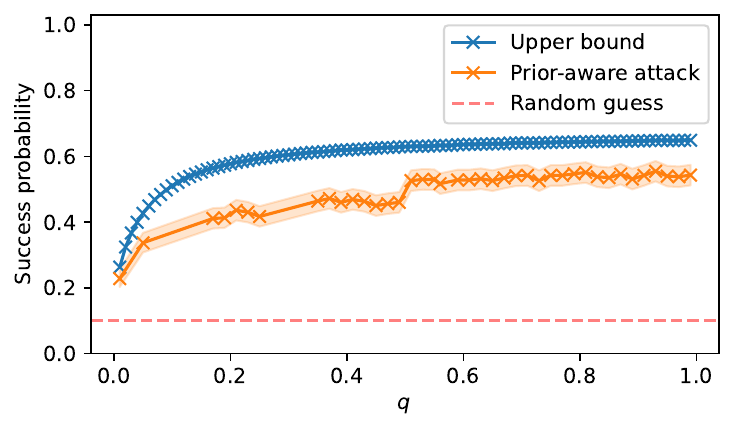}
        \caption{$\epsilon=8, |\pi|=10$.}
\end{subfigure}%
\begin{subfigure}{0.25\textwidth}
\centering
    \includegraphics[width=\textwidth]{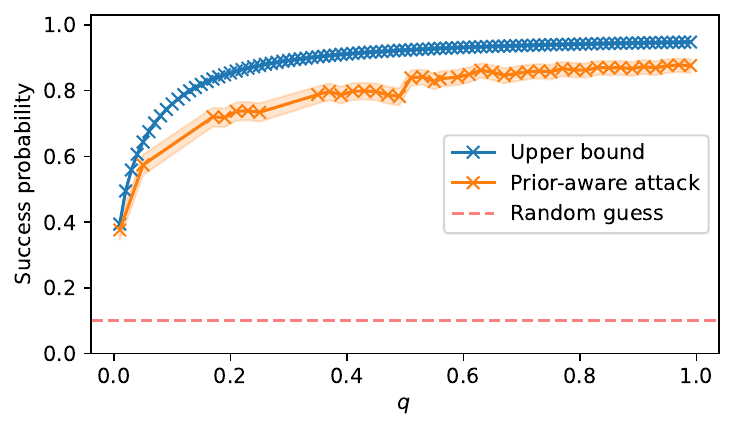}
        \caption{$\epsilon=16, |\pi|=10$.}
\end{subfigure}

\begin{subfigure}{0.25\textwidth}
\centering
    \includegraphics[width=\textwidth]{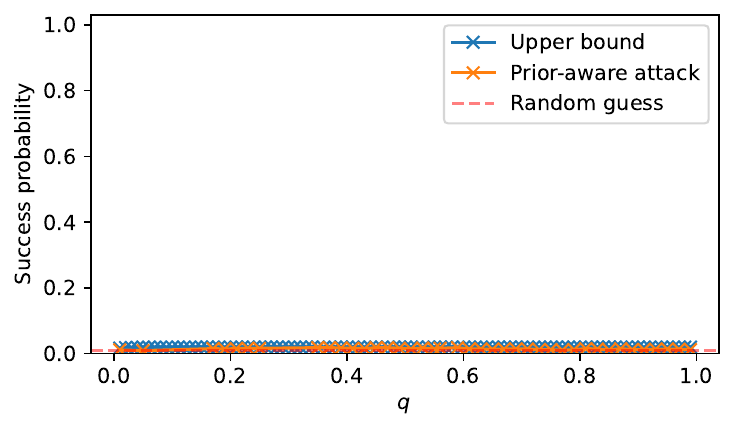}
        \caption{$\epsilon=1, |\pi|=100$.}
\end{subfigure}%
\begin{subfigure}{0.25\textwidth}
\centering
    \includegraphics[width=\textwidth]{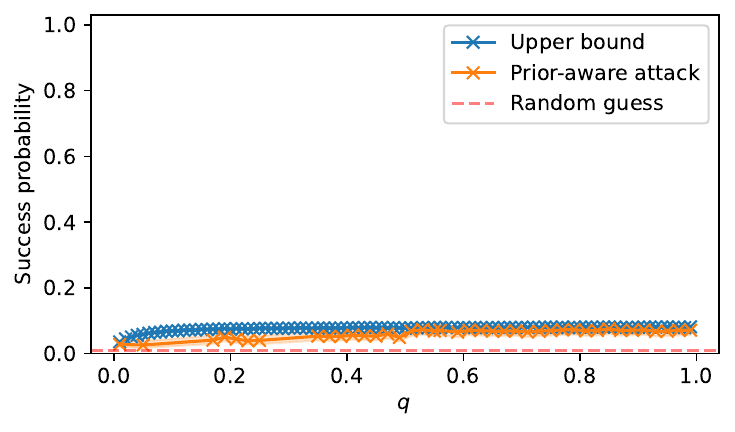}
       \caption{$\epsilon=4, |\pi|=100$.}
\end{subfigure}%
\begin{subfigure}{0.25\textwidth}
\centering
    \includegraphics[width=\textwidth]{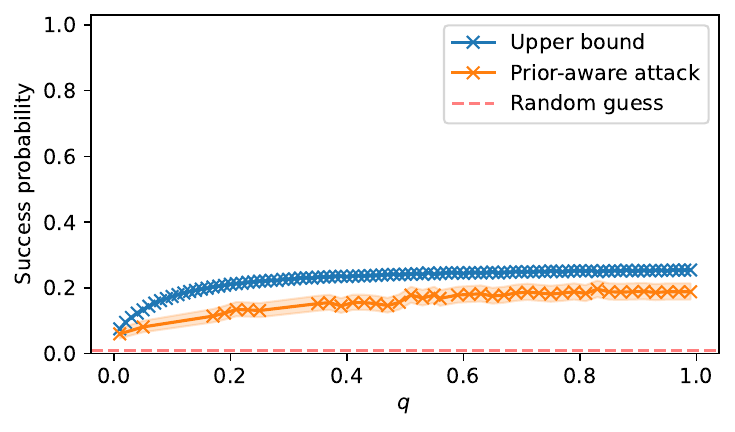}
        \caption{$\epsilon=8, |\pi|=100$.}
\end{subfigure}%
\begin{subfigure}{0.25\textwidth}
\centering
    \includegraphics[width=\textwidth]{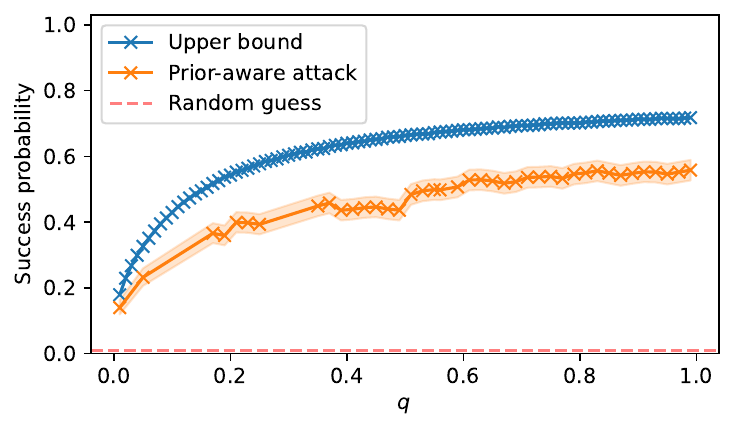}
        \caption{$\epsilon=16, |\pi|=100$.}
\end{subfigure}%

\caption{How the upper bound and (improved) prior-aware attack change as a function of $q$ at a fixed value of $\epsilon$ and prior size, $|\pi|$. The amount of privacy leaked through a reconstruction at a fixed value of $\epsilon$ can change with different $q$.}
\label{fig: compare_dp_hyps_upper_and_lower}
\end{figure}

\begin{figure*}[t]
  \centering

\begin{subfigure}{0.32\textwidth}
\centering
    \includegraphics[width=\textwidth]{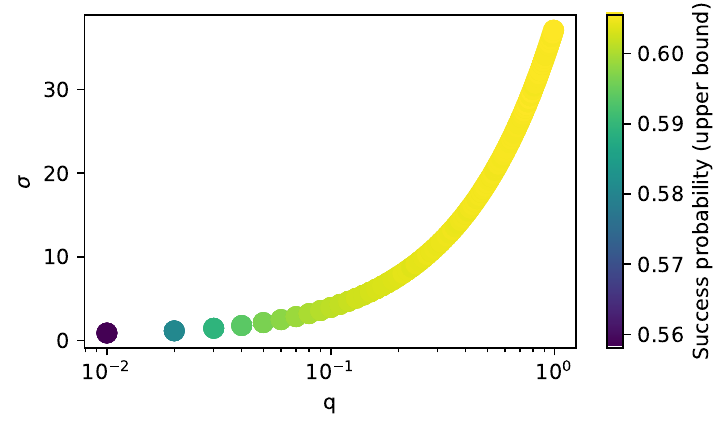}
        \caption{$\epsilon=1, |\pi|=2$.}
\end{subfigure}%
\begin{subfigure}{0.32\textwidth}
\centering
    \includegraphics[width=\textwidth]{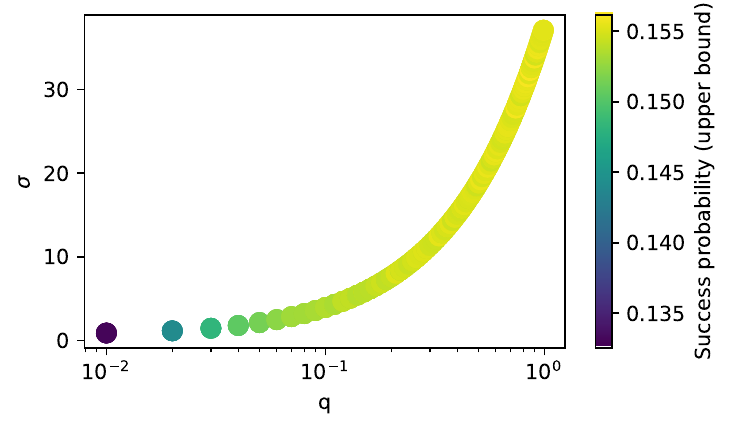}
        \caption{$\epsilon=1, |\pi|=10$.}
\end{subfigure}%
\begin{subfigure}{0.32\textwidth}
\centering
    \includegraphics[width=\textwidth]{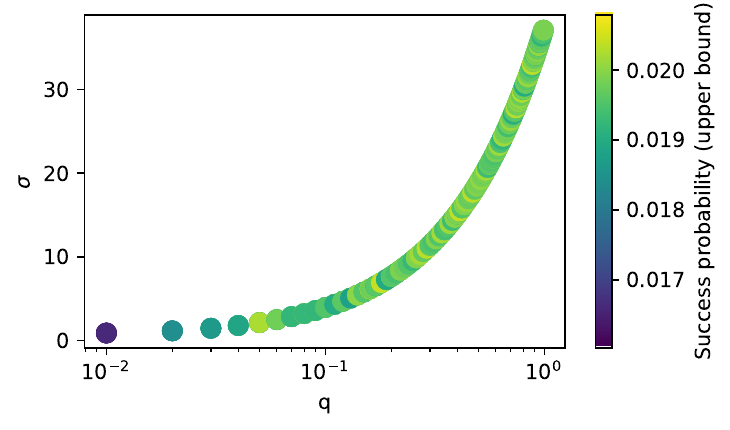}
        \caption{$\epsilon=1, |\pi|=100$.}
\end{subfigure}

\begin{subfigure}{0.32\textwidth}
\centering
    \includegraphics[width=\textwidth]{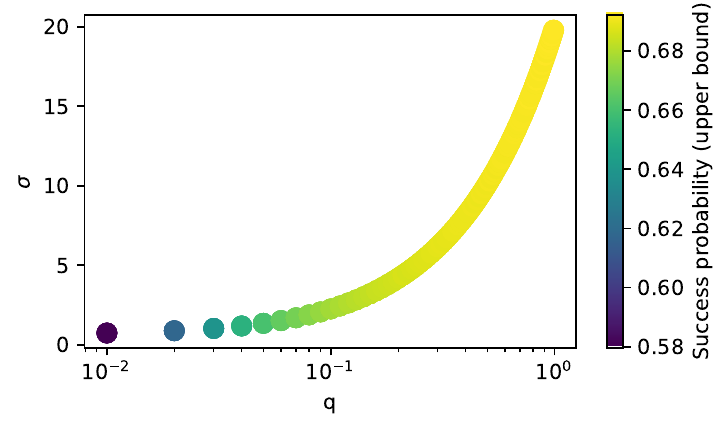}
        \caption{$\epsilon=2, |\pi|=2$.}
\end{subfigure}%
\begin{subfigure}{0.32\textwidth}
\centering
    \includegraphics[width=\textwidth]{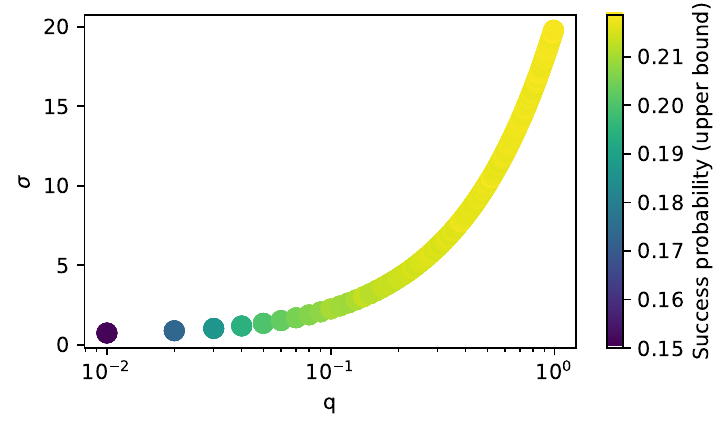}
        \caption{$\epsilon=2, |\pi|=10$.}
\end{subfigure}%
\begin{subfigure}{0.32\textwidth}
\centering
    \includegraphics[width=\textwidth]{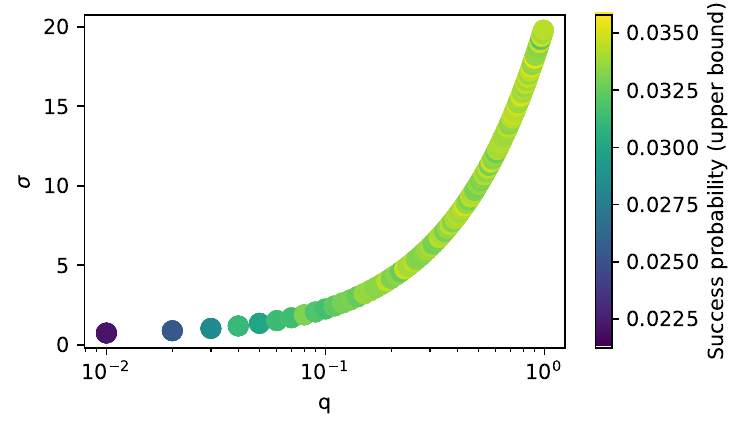}
        \caption{$\epsilon=2, |\pi|=100$.}
\end{subfigure}

\begin{subfigure}{0.32\textwidth}
\centering
    \includegraphics[width=\textwidth]{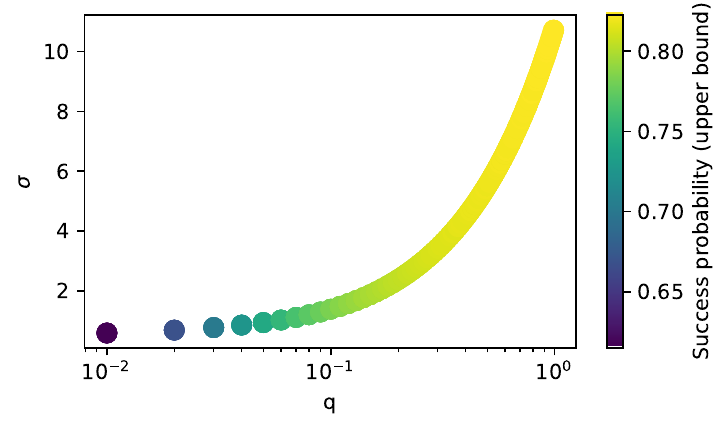}
        \caption{$\epsilon=4, |\pi|=2$.}
\end{subfigure}%
\begin{subfigure}{0.32\textwidth}
\centering
    \includegraphics[width=\textwidth]{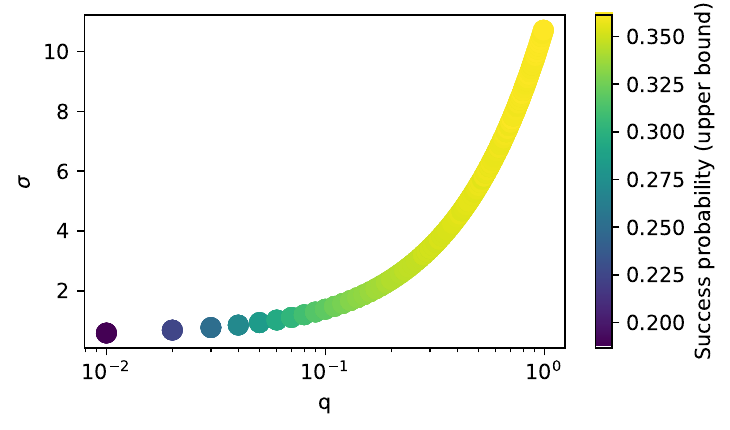}
        \caption{$\epsilon=4, |\pi|=10$.}
\end{subfigure}%
\begin{subfigure}{0.32\textwidth}
\centering
    \includegraphics[width=\textwidth]{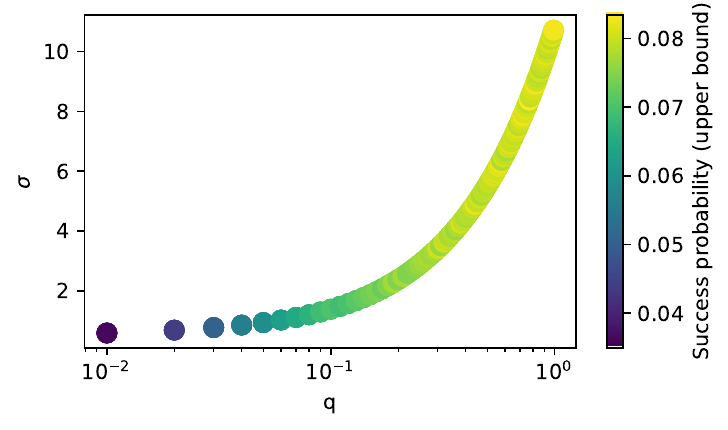}
        \caption{$\epsilon=4, |\pi|=100$.}
\end{subfigure}

\begin{subfigure}{0.32\textwidth}
\centering
    \includegraphics[width=\textwidth]{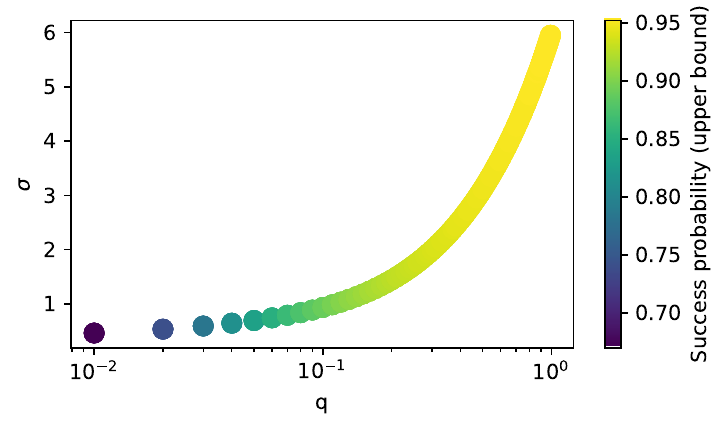}
        \caption{$\epsilon=8, |\pi|=2$.}
\end{subfigure}%
\begin{subfigure}{0.32\textwidth}
\centering
    \includegraphics[width=\textwidth]{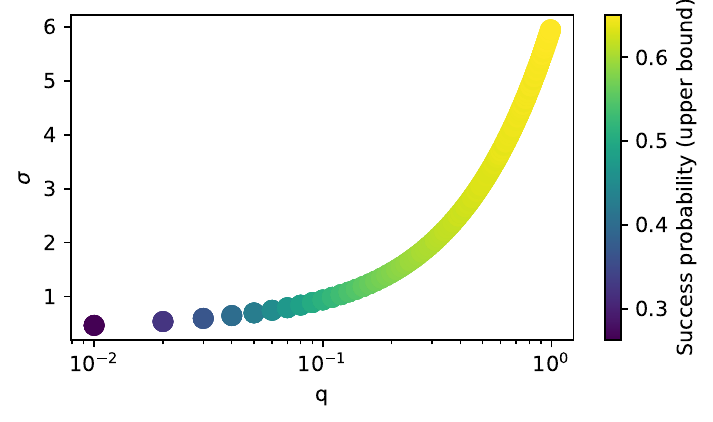}
        \caption{$\epsilon=8, |\pi|=10$.}
\end{subfigure}%
\begin{subfigure}{0.32\textwidth}
\centering
    \includegraphics[width=\textwidth]{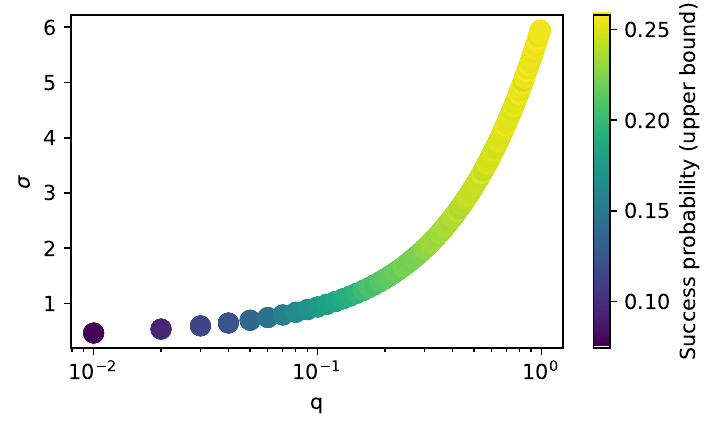}
        \caption{$\epsilon=8, |\pi|=100$.}
\end{subfigure}

\begin{subfigure}{0.32\textwidth}
\centering
    \includegraphics[width=\textwidth]{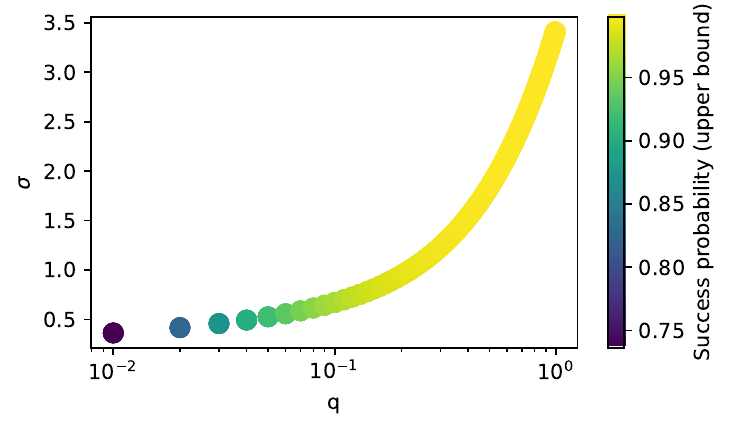}
        \caption{$\epsilon=16, |\pi|=2$.}
\end{subfigure}%
\begin{subfigure}{0.32\textwidth}
\centering
    \includegraphics[width=\textwidth]{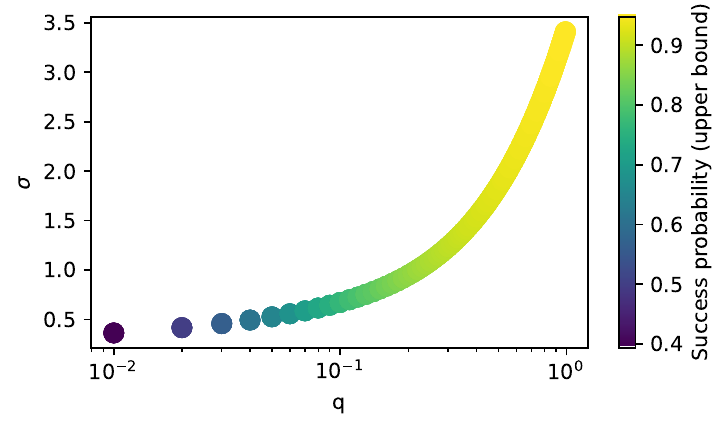}
        \caption{$\epsilon=16, |\pi|=10$.}
\end{subfigure}%
\begin{subfigure}{0.32\textwidth}
\centering
    \includegraphics[width=\textwidth]{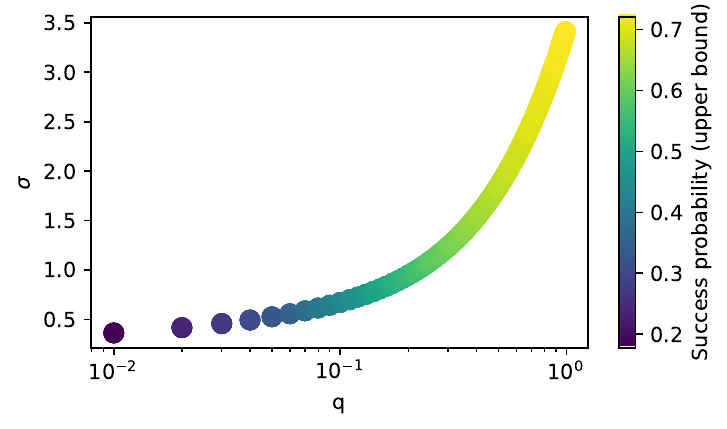}
        \caption{$\epsilon=16, |\pi|=100$.}
\end{subfigure}

\begin{subfigure}{0.32\textwidth}
\centering
    \includegraphics[width=\textwidth]{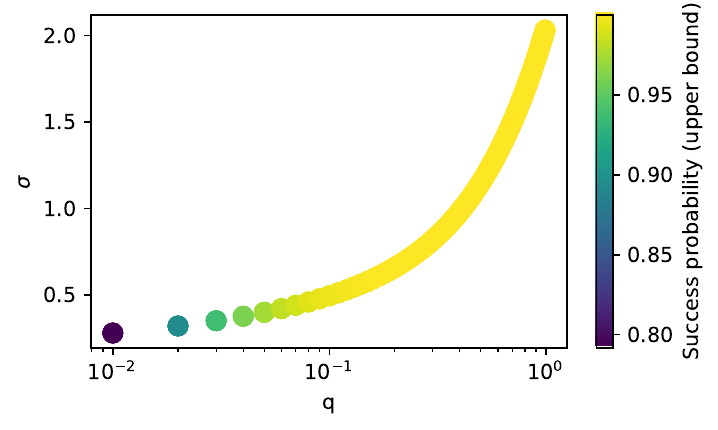}
        \caption{$\epsilon=32, |\pi|=2$.}
\end{subfigure}%
\begin{subfigure}{0.32\textwidth}
\centering
    \includegraphics[width=\textwidth]{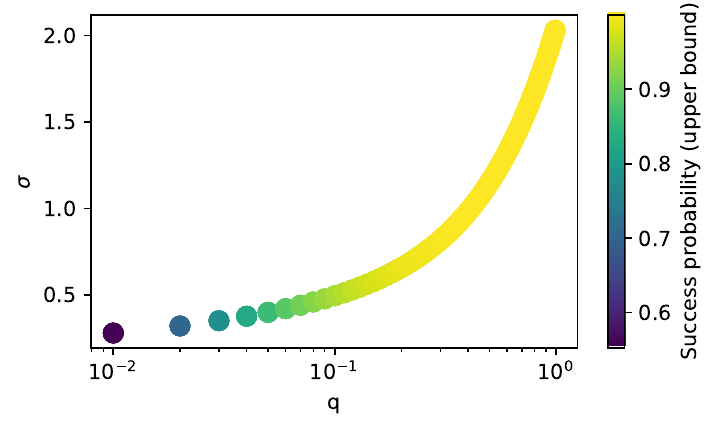}
        \caption{$\epsilon=32, |\pi|=10$.}
\end{subfigure}%
\begin{subfigure}{0.32\textwidth}
\centering
    \includegraphics[width=\textwidth]{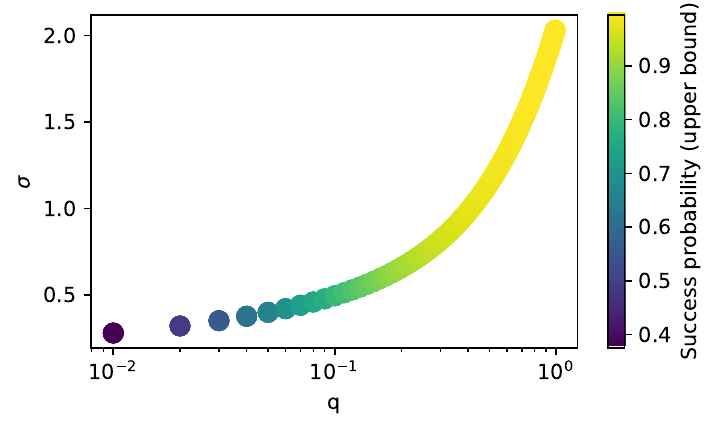}
        \caption{$\epsilon=32, |\pi|=100$.}
\end{subfigure}

\caption{How the upper bound changes as a function of $q$ and $\sigma$ at a fixed value of $\epsilon$ and prior size, $|\pi|$, and setting $T=100$. The probability of a successful reconstruction can vary widely with different values of $q$. For example, at $\epsilon=32$ and $|\pi|=100$, at $q=0.01$ the upper bound is 0.4 and at $q=0.99$ it is 1. Note that the color ranges are independent across subfigures.}
\label{fig: compare_dp_hyps_upper}
\end{figure*}
\section{Estimating $\kappa$ from samples}\label{app:estimate_kappa}
Here, we discuss how to estimate the base probability of reconstruction success, $\kappa$, if the adversary can only sample from the prior distribution.

Let $\hat{\pi}$ be the empirical distribution obtained by taking $N$ independent samples from the prior and $\hat{\kappa} = \kappa_{\hat{\pi},\lattackloss}(\eta)$ be the corresponding parameter for this discrete approximation to $\pi$ -- this can be computed using the methods sketched in \Cref{sec:bounding_reconstruction}.
Then we have the following concentration bound.

\begin{proposition}\label{prop:kappa_hat}
With probability $1 - e^{-N \tau^2 \kappa /2}$ we have
\begin{align*}
    \kappa \leq \frac{\hat{\kappa}}{1 - \tau} \enspace.
\end{align*}
\end{proposition}

The proof is given in \Cref{app:proofs}.

\section{Proofs}\label{app:proofs}

Throughout the proofs we make some of our notation more succinct for convenience.
For a probability distribution $\omega$ we write $\omega(E) = \Pr_{\omega}[E]$, and rewrite
$\blowup_{\kappa}(\mu, \nu) = \sup\{\Pr_{\mu}[E] : E \;\; \text{s.t.} \;\; \Pr_{\nu}[E] \leq \kappa \}$ as $\sup_{\nu(E) \leq \kappa} \mu(E)$.
Given a distribution $\omega$ and function $\phi$ taking values in $[0,1]$ we also write $\omega(\phi) = \Ex_{X \sim \omega}[\phi(X)]$.

\subsection{Proof of \Cref{thm:bound_reconstruction}}\label{app:proof_bound_reconstruction}

We say that a pair of distributions $(\mu, \nu)$ is \emph{testable} if for all $\kappa \in [0,1]$ we have
\begin{align*}
    \inf_{\nu(\phi) \leq \kappa} (1- \mu(\phi)) = \inf_{\nu(E) \leq \kappa} (1-\mu(E)) \enspace,
\end{align*}
where the infimum on the left is over all $[0,1]$-valued measurable functions and the one on the right is over measurable events (i.e.\ $\{0,1\}$-valued functions).
The Neyman-Pearson lemma (see e.g.\ \cite{lehmann2005testing}) implies that this condition is satisfied whenever the statistical hypothesis problem of distinguishing between $\mu$ and $\nu$ admits a uniformly most powerful test.
For example, this is the case for distributions on $\R^d$ where the density ratio $\mu / \nu$ is a continuous function.

\begin{theorem}[Formal version of \Cref{thm:bound_reconstruction}]\label{thm:bound_reconstruction_full}
Fix $\pi$ and $\lattackloss$.
Suppose that for every fixed dataset $\Dfixed$ there exists a distribution $\mu_{\Dfixed}$ such that $\sup_{z \in \supp(\pi)} \blowup_{\kappa}(\mu_{\Dz}, \nu_{\Dfixed}) \leq \blowup_{\kappa}(\mu_{\Dfixed}, \nu_{\Dfixed})$ for all $\kappa \in [0,1]$.
If the pair $(\mu, \nu)$ is testable, then $M$ is $(\eta, \gamma)$-ReRo with
\begin{align*}
    \gamma = \sup_{\Dfixed} \sup_{\nu_{\Dfixed}(E) \leq \kappa_{\pi,\lattackloss}(\eta)} \mu_{\Dfixed}(E) \enspace.
\end{align*}
\end{theorem}

The following lemma from \citet{dong2019gaussian} will be useful.

\begin{lemma}\label{lem:tradeoff_convex}
For any $\mu$ and $\nu$, the function $\kappa \mapsto \inf_{\nu(\phi) \leq \kappa} (1 - \mu(\phi))$ is convex in $[0,1]$.
\end{lemma}

\begin{lemma}\label{lem:gamma_concave}
For any testable pair $(\mu, \nu)$, the function $\kappa \mapsto \sup_{\nu(E) \leq \kappa} \mu(E)$ is concave.
\end{lemma}
\begin{proof}
By the testability assumption we have
\begin{align*}
    \sup_{\nu(E) \leq \kappa} \mu(E)
    &=
    \sup_{\nu(E) \leq \kappa} \mu(E)
    \\
    &=
    \sup_{\nu(E) \leq \kappa} (1 - \mu(\bar{E}))
    \\
    &=
    1 - \inf_{\nu(E) \leq \kappa} \mu(\bar{E})
    \\
    &=
    1 - \inf_{\nu(E) \leq \kappa} (1 - \mu(E))
    \\
    &=
    1 - \inf_{\nu(\phi) \leq \kappa} (1 - \mu(\phi)) \enspace.
\end{align*}
Concavity now follows from \Cref{lem:tradeoff_convex}.
\end{proof}

\begin{proof}[Proof of \Cref{thm:bound_reconstruction_full}]
Fix $\Dfixed$ and let $\kappa = \kappa_{\pi,\lattackloss}(\eta)$ throughout.
Let also $\nu = \nu_{\Dfixed}$, $\mu_z = \mu_{\Dz}$, $\nu^* = \nu^*_{\Dfixed}$ and $\mu^* = \mu^*_{\Dfixed}$.

Expanding the probability of successful reconstruction, we get:
\begin{align*}
    \Pr_{Z \sim \pi, W \sim M(\Dfixed \cup \{Z\})}[\lattackloss(Z, \reconstruct(W)) \leq \eta]
    &=
    \Ex_{Z \sim \pi} \Pr_{W \sim M(\Dfixed \cup \{Z\})}[\lattackloss(Z, \reconstruct(W)) \leq \eta]
    \\
    &=
    \Ex_{Z \sim \pi} \Ex_{W \sim M(\Dfixed \cup \{Z\})}\indicator[\lattackloss(Z, \reconstruct(W)) \leq \eta]
    \\
    &=
    \Ex_{Z \sim \pi} \Ex_{W \sim \mu_Z}\indicator[\lattackloss(Z, \reconstruct(W)) \leq \eta]
    \\
    &=
    \Ex_{Z \sim \pi} \Ex_{W \sim \nu}\left[\frac{\mu_Z(W)}{\nu(W)} \indicator[\lattackloss(Z, \reconstruct(w)) \leq \eta]\right] \enspace.
\end{align*}

Now fix $z \in \supp(\pi)$ and let $\kappa_z = \Pr_{W \sim \nu}[\lattackloss(z, \reconstruct(W)) \leq \eta]$.
Using the assumption on $\mu^*$ we get:

\begin{align*}
    \Ex_{W \sim \nu}\left[\frac{\mu_z(W)}{\nu(W)} \indicator[\lattackloss(z, \reconstruct(w)) \leq \eta]\right]
    &\leq
    \sup_{\nu(E) \leq \kappa_z} \Ex_{W \sim \nu}\left[\frac{\mu_z(W)}{\nu(W)} \indicator[W \in E]\right]
     &\text{(By definition of $\kappa$)}\\
    &=
    \sup_{\nu(E) \leq \kappa_z} \Ex_{W \sim \mu_z}\left[\indicator[W \in E]\right]
    \\
    &=
    \sup_{\nu(E) \leq \kappa_z} \mu_z(E)
    \\
    &\leq
    \sup_{\nu^*(E) \leq \kappa_z} \mu^*(E)
    \enspace. &\text{(By definition of $\mu^*$ and $\nu^*$.)}
\end{align*}

Finally, using \Cref{lem:gamma_concave} and Jensen's inequality on the following gives the result:
\begin{align*}
    \Ex_{Z \sim \pi}[ \kappa_Z ]
    &=
    \Ex_{Z \sim \pi} \Pr_{W \sim \nu}[\lattackloss(Z, \reconstruct(W)) \leq \eta]
    \\
    &=
    \Ex_{W \sim \nu} \Pr_{Z \sim \pi}[\lattackloss(Z, \reconstruct(W)) \leq \eta]
    \\
    &\leq
    \Ex_{W \sim \nu} \kappa
    \\
    &=
    \kappa \enspace. \qedhere
\end{align*}
\end{proof}

\subsection{Proof of \Cref{cor:bound_reconstruction_dpsgd}}
Here we prove Corollary \ref{cor:bound_reconstruction_dpsgd}.
We will use the following shorthand notation for convenience: $\mu=\cN(B(T,q), \sigma^2 I)$ and $\nu = \cN(0, \sigma^2 I)$. To prove our result, we use the notion of $TV_a$.

\begin{definition}[\citet{DBLP:journals/corr/abs-2204-06106}]
For two probability distributions $\omega_1(\cdot)$ and $\omega_2(\cdot)$, $TV_a$ is defined as
$$TV_a(\omega_1,\omega_2) = \int \left|\omega_1(x) - a\cdot \omega_2(x)\right|dx.$$
\end{definition}

Now we state the following lemma borrowed from \cite{DBLP:journals/corr/abs-2204-06106}. 

\begin{lemma}[Theorem 6 in \cite{DBLP:journals/corr/abs-2204-06106}]
Let $\nu_{\Dfixed}$, $\mu_{\Dz}$ be the output distribution of DP-SGD applied to $\Dfixed$ and $\Dz$ respectively, with noise multiplier $\sigma$, sampling rate $q$. Then we have
\begin{align*}
TV_a(\nu_{\Dfixed},\mu_{\Dz}) \leq TV_a(\nu , \mu) \enspace.    
\end{align*}
\end{lemma}

Now, we state the following lemma that connects $TV_a$ to blow-up function.
\begin{lemma}[Lemma 21 in \cite{zhu2022optimal}.]
For any pair of distributions $\omega_1, \omega_2$ we have
$$\sup_{\omega_1(E)\leq \kappa } \omega_2(E)= \inf_{a>1} \min\left\{0, a\cdot \kappa + \frac{TV_a(\omega_1, \omega_2) + 1 - a}{2}, \frac{2\kappa + TV_a(\omega_1,\omega_2) + a-1}{2a} \right\}$$
\end{lemma}

Since $TV_a(\nu_{\Dfixed},\mu_{\Dz})$ is bounded by $TV_a(\nu, \mu)$ for all $a$, therefore we have 
\begin{align*}
\sup_{\nu_{\Dfixed}(E)\leq \kappa } \mu_{\Dz}(E) \leq \sup_{\nu(E)\leq \kappa} \mu(E) \enspace.    
\end{align*}

\subsection{Proof of \Cref{prop:kappa_hat}}
Recall $\kappa = \sup_{z_0 \in \Zset} \Pr_{Z \sim \pi}[\lattackloss(Z, z_0) \leq \eta]$ and $\hat{\kappa} = \sup_{z_0 \in \Zset} \Pr_{Z \sim \hat{\pi}}[\lattackloss(Z, z_0) \leq \eta]$.
Let $\kappa_z = \Pr_{Z \sim \pi}[\lattackloss(Z, z) \leq \eta]$ and $\hat{\kappa}_z = \Pr_{Z \sim \hat{\pi}}[\lattackloss(Z, z) \leq \eta]$.
Note $\hat{\kappa}_z$ is the sum of $N$ i.i.d.\ Bernoulli random variables and $\Ex_{\hat{\pi}}[\hat{\kappa}_z] = \kappa_z$.
Then, using a multiplicative Chernoff bound,
we see that for a fixed $z$ the following holds with probability at least $1 - e^{-N \tau_z^2 \kappa /2}$:
\begin{align*}
    \kappa_z \leq \frac{\hat{\kappa}_z}{1-\tau} \enspace.
\end{align*}
Applying this to $z^* = \argsup_{z_0 \in \Zset} \Pr_{Z \sim \pi}[\lattackloss(Z, z_0) \leq \eta]$ we get that the following holds with probability at least $1 - e^{-N \tau^2 \kappa /2}$:
\begin{align*}
    \kappa &=
    \kappa_{z^*}
    \leq
    \frac{\hat{\kappa}_{z^*}}{1-\tau}
    \leq
    \frac{\hat{\kappa}}{1-\tau} \enspace.
\end{align*}

\subsection{Proof of \Cref{prop:gamma_consistent}}

Let $z= \frac{(r'_{N'} + r'_{N'-1})}{2}$. Let $E_1$ be the event that $|\Pr[r>z] - \kappa|\geq \tau$. By applying Chernoff-Hoefding bound we have $\Pr[E_1]\leq 2e^{-2N\tau^2}$.
Now note that since $\mu$ is a Gaussian mixture, we can write $\mu=\sum_{i\in [2^T]} a_i \mu_i$ where each $\mu_i$ is a Gaussian $\cN(c_i,\sigma)$ where $|c_i|_2\leq \sqrt{T}$. Now let $r_i = \mu_i(W)/\nu(W)$. By holder, we have $\Ex[r^2]\leq \sum a_i \Ex[r_i^2]$. We also now that $\Ex[r_i^2]\leq e^T$, therefore, $\Ex[r^2]\leq e^T.$ Now let $E_2$ be the event that $|\Ex[r\cdot I(r>z)] -\gamma'| \geq  \tau.$ Since the second moment of $r$ is bounded, the probability of $E_2$ goes to zero as $N$ increases.  Therefore, almost surely we have
$$ \sup_{\nu(E)\leq \kappa - \tau} \mu(E) -\tau \leq \lim_{N\to \infty} \gamma' \leq \sup_{\nu(E)\leq \kappa +\tau} \mu(E) +\tau.$$
Now by pushing $\tau$ to $0$ and using the fact that $\mu$ and $\nu$ are smooth we have
$$\lim_{N\to \infty} \gamma' = \lim_{\tau\to0}\sup_{\nu(E)\leq \kappa +\tau} \mu(E) +\tau=\sup_{\nu(E)\leq \kappa} \mu(E).$$

\end{document}